\documentclass[10pt,journal,compsoc]{IEEEtran}
\ifCLASSOPTIONcompsoc
  \usepackage[nocompress]{cite}
\else
  \usepackage{cite}
\fi
\ifCLASSINFOpdf
   \usepackage[pdftex]{graphicx}
\else
  \usepackage[dvips]{graphicx}
\fi
\usepackage{amsthm} 

\newtheorem{Mytheo}{Theorem}
\usepackage{multirow}
\usepackage{booktabs}

\usepackage{algorithm}
\usepackage{algorithmic}

\begin{document}
\title{CHAOS: Accurate and Realtime Detection of Aging-Oriented Failure Using Entropy}

\author{Pengfei~Chen,%~\IEEEmembership{Member,~IEEE,}
        Yong~Qi,~\IEEEmembership{Member,~IEEE,}%~\IEEEmembership{Fellow,~OSA,}
        and~Di~Hou,%~\IEEEmembership{Life~Fellow,~IEEE}% <-this % stops a space
\IEEEcompsocitemizethanks{\IEEEcompsocthanksitem P.Chen, Y.Qi and D.Hou are with the Department
of Computer Science and Technology, Xi'an Jiaotong University, Xi'an,
China, 710049.\protect\\
E-mail: chenpengfei@outlook.com }}
%\IEEEcompsocthanksitem J. Doe and J. Doe are with Anonymous University.}}% <-this % stops a space
%\thanks{Manuscript received April 19, 2005; revised January 17, 2015.}}

\markboth{IEEE TRANSACTIONS ON COMPUTERS, January~2015}%
{Shell \MakeLowercase{\textit{et al.}}: Bare Advanced Demo of IEEEtran.cls for Journals}

%\IEEEpubid{\makebox[\columnwidth]{\hfill 0000--0000/00/\$00.00~\copyright~2014 IEEE}%


\IEEEtitleabstractindextext{%
\begin{abstract}
%%% These chronic problems often fly under the radar of software monitoring system because they do not significantly affect software state. 
%High availability and  performance are  essential concerns %for both of small and
%for long running software systems. 
%However, 
Even well-designed software systems suffer from chronic performance degradation, also named ``software aging", due to  internal (e.g. software bugs) and external (e.g. resource exhaustion) impairments.
These chronic problems often fly under the radar of software monitoring systems before causing severe impacts (e.g. system failure). Therefore it's a challenging issue how to timely detect these problems to prevent system crash.  Although a large quantity of approaches have been proposed to solve this issue, the accuracy and effectiveness of these approaches are still  far from satisfactory due to the  insufficiency of aging indicators adopted by them.  
%The chronic performance degradation  may gradually evolve to be severe problems (e.g. system failure).  
%Hence it should be detected accurately in real time in order to conduct some ``rejuvenation" strategies. 
%To provide  accurate approaches
%To accurately detect failure-prone state caused by software aging, 
In this paper, we present a novel  entropy-based aging indicator, Multidimensional Multi-scale Entropy (MMSE). MMSE employs the complexity embedded in  runtime performance metrics to indicate software aging and leverages multi-scale and multi-dimension integration to tolerate system fluctuations. Via theoretical proof and experimental evaluation, we demonstrate that MMSE satisfies \textit{Stability}, \textit{Monotonicity} and \textit{Integration} which we conjecture that an ideal  aging indicator should have. Based upon MMSE, we develop three failure detection approaches encapsulated in a proof-of-concept named CHAOS. 
%entropy based detection framework named CHAOS. The fundamental prerequisite of CHAOS is that entropy of the runtime system cloud be leveraged as an efficient indicator of software aging. %Based upon this prerequisite, CHAOS  extends the conventional  Multi-scale Entropy (MSE) to  Multidimensional Multi-scale Entropy (MMSE) which satisfies  the three properties: \textit{Stability}, \textit{Monotonicity} and \textit{Integration} which we conjecture that an ideal  aging indicator should have and encompasses three realtime failure detection approaches. 
The experimental evaluations in a Video on Demand (VoD) system %deployed in a controlled environment 
and in a real-world  production system, AntVision, show that CHAOS can detect the failure-prone state  in  an extraordinarily high accuracy and a near 0  \textit{Ahead-Time-To-Failure ($ATTF$)}. Compared to previous approaches, %based upon  explicit aging indicators (e.g. CPU utilization) and a state-of-the-art approach based upon implicit aging indicator i.e. $H\ddot{o}lder$ exponent,
 CHAOS improves the detection  accuracy  by about 5 times and reduces the $ATTF$ even by 3 orders of magnitude. %We believe that the high accuracy of CHAOS are of essence to conduct cost-effective software rejuvenation. 
In addition,  CHAOS is light-weight enough to satisfy the realtime requirement.  
\end{abstract}

}

\maketitle

% is not selected <OR> if conference mode is selected - because compsoc
% conference papers position the abstract like regular (non-compsoc)
\IEEEdisplaynontitleabstractindextext

% \fi
\IEEEpeerreviewmaketitle

\ifCLASSOPTIONcompsoc
\IEEEraisesectionheading{\section{Introduction}\label{sec:introduction}}
\else
\section{Introduction}
\label{sec:introduction}
\fi
Software is becoming the backbone of modern society. Especially with the development of cloud computing, more and more traditional services (e.g. food ordering, retail) %from all aspects of our daily life
are deployed in the cloud and function as distributed software systems. Two common characteristics of those software systems, namely long-running and high complexity increase the risks of faults  and resource exhaustion. With the accumulation of faults or resource consumption, software systems may suffer from chronic performance degradation, failure rate/probability increase and even crash called ``software aging" \cite{0,1,2,3,4} or ``Chronics" \cite{5}. %in previous studies. 

Software aging %or performance deterioration
has been extensively studied for two decades since it was first quantitatively analyzed in AT\&T lab in 1995 \cite{6}.
%According to previous literatures, 
This phenomenon has been widely observed in variant software systems nearly spanning across all software stacks such as cloud computing infrastructure (e.g. Eucalyptus) \cite{7,8}, virtual machine monitor (VMM)\cite{10,11}, operating system\cite{0,12}, Java Virtual Machine (JVM) \cite{4,13}, web server \cite{3,14} and so on. %Although the fundamental mechanism of software aging keeps still a mystery, some common causes are widely accepted such as memory leaking,  round-off error accumulation, data corruption and so on.
As the degree of software aging increasing, software performance  decreases gradually resulting in QoS (e.g. response time) decrease. %When the performance exceeds a preset threshold, the software system steps into a failure-prone or failure state. For instance, for a web server  when the request response time exceeds a preset threshold (e.g. 20 seconds), we say the system is unable to provide service normally and a ``failure" occurs although it is still alive. 
 What's worse, software aging may lead to  unplanned system hang or crash %and degrade the availability significantly. 
The unplanned outage in enterprise  system especially in cloud platform can cause considerable revenue loss. A recent survey shows that IT downtime on an average leads to 14 hours of downtime per year, leading to \$26.5 billion lost \cite{15}. %And it's also possible that the unplanned outage threatens people's life such as the well-known accident of the Patriot anti-missile system that caused the loss of human lives \cite{16}. 
Therefore detecting and counteracting software aging are of essence for building long-running systems. 

An efficient and commonly used counteracting software aging strategy is ``software rejuvenation" \cite{2,3,4,17}, which proactively recovers the system from failure-prone state to a completely or partially new state  by cleaning the internal state. The benefit of rejuvenation strategies heavily depends on the time triggering rejuvenation. Frequent rejuvenation actions may decrease the system availability  or performance due to the non-ignorable planed downtime or overhead caused by such actions. Instead, an ideal rejuvenation strategy is to recover the system when it just gets near to the failure-prone state.% Hence this paper is mainly concerned with the failure-prone state  rather than the gradual degradation procedure. 
We name the failure-prone state caused by software aging as ``Aging-Oriented Failure"(AOF). Different from  transient failures caused by fatal errors e.g. segment fault or hardware failures, AOF is a kind of ``chronics" \cite{5} which means some durable anomalies  have emerged before system crash. Therefore AOF is likely to be detected. Accurately detecting AOF is a critical problem and the goal of this paper. However, to that end, we confront the following three challenges: 
\begin{itemize}
\item Different from fail-stop problems e.g. crash or hang  which have sufficient and observable indicators (e.g. exceptions), 
%For crash or hang  problems, there are sufficient and obvious symptoms (e.g.segment faults) to indicate such problems. However, for
non-crash failures caused by software aging where the server does not crash but fails to process the request compliant with the SLA constraints, have no observable and sufficient symptoms to indicate them. These failures often fly under the radar of monitoring systems. Hence, finding out the underlying  indicator for software aging becomes the first challenge.  
\item The internal state (e.g. memory leak) changes and external state (e.g. workload variation) changes make the running system extraordinarily complex. Hence, the running system may not be described neither by a simple linear model nor by a single performance metric. How to cover the complexity and multi-dimension in the aging indicator is the second challenge.
\item Fluctuations or noise may be involved in collected performance metrics due to the highly dynamic property of the running system. And cloud computing exacerbates the dynamics due to its elasticity and flexibility  (e.g. VM creation and deletion). How to mitigate the influence of noise and keep the detection approach noise-resilient is the third challenge.   
\end{itemize}
 
%A wide spectrum of papers have proposed enormous methods to detect or predict software aging. According to the survey \cite{18} , these methods could be roughly categorized into three groups: time series \cite{0,3,8,12,19,20}, machine learning \cite{4,21,50} and threshold-based approach \cite{22,23}. These methods may accurately detect or predict the  explicitly observed runtime performance metrics (e.g. CPU utilization) which are usually leveraged as aging indicators. However, the detection or prediction results may be biased due to the insufficiency of these indicators especially when the system fluctuates a lot.  %Even an irony exists: the more accurate the detection or prediction is, the worse the effectiveness is. 
%The insufficiency exists as these runtime parameters are only the external observations rather than the mapping functions which could map the observations to the system internal states which will be illustrated in Section \uppercase\expandafter{\romannumeral 2}.
%these traditional aging indicators will be depicted detailly in Section \uppercase\expandafter{\romannumeral 3}.

To address the aforementioned challenges, %and shortages of current approaches, 
we conjecture that an ideal aging indicator should have \textit{Monotonicity} property to reveal the hidden aging state, \textit{Integration} property to comprehensively describe aging process and \textit{Stability} property to tolerate system fluctuation. In this paper, we propose a novel  %entropy \footnote{If there is no special statement, ``Entropy" means ``Shanon Entropy" in this paper.} based 
aging indicator named MMSE. According to our observation in practice and qualitative proof,  entropy monotonously increases with the degree of software aging when the failure probability is lower than 0.5.  %Hence, MMSE is capable of reflecting the hidden software aging state.
And MMSE is a complexity oriented and model-free indicator without deterministic linear or non-linear model assumptions. In addition, the multi-scale  feature mitigates the influence of system fluctuations and the multi-dimension feature makes MMSE more comprehensive to describe software aging. Hence, MMSE satisfies the three properties namely \textit{Stability}, \textit{Monotonicity} and \textit{Integration}, which we conjecture that an ideal aging indicator should have. Based upon MMSE, we develop three AOF detection approaches encapsulated in a proof-of-concept, CHAOS. 
 %MMSE is obtained by extending Multi-scale Entropy (MSE) via several modifications 
%to detect the failure-prone state caused by software aging. 
%This work is first inspired by the research results in  \cite{24} and  \cite{14}.  In \cite{24},  Mark, et.al studied the multifractality of memory resources using $H\ddot{o}lder$ exponent and illustrated  that the  $H\ddot{o}lder$ exponent decreased significantly during software aging.  In \cite{14}, Jia, et.al proposed a nonlinear approach to modeling of software aging. A common point implied in their work is the complexity increases  with the degree of software aging . Hence, an intuitional speculation  is  entropy as a measurement of complexity may have the potential to be an aging indicator. 
%Via a quantitative proof, we prove that entropy increases with the degree of software aging when the failure probability is lower than 0.5. In CHAOS, we extend the traditional MSE  to  MMSE and demonstrate that MMSE satisfies the three properties: \textit{Stability}, \textit{Monotonicity} and \textit{Integration} which we conjecture that an ideal aging indicator should have. Based upon  MMSE, we develop three realtime failure detection methods. And 
To further decrease the overhead caused by CHAOS, we reduce the runtime performance metrics  from 76 to 5 without significant information loss by a principal component analysis (PCA) based variable selection method. The experimental evaluations in a VoD system %deployed in a controlled environment
 and in a real production system, AntVision \footnote{www.antvision.net},  show that CHAOS  has a strong power to detect failure-prone state  with a high accuracy and a small  $ATTF$. Compared to  precious approaches %based upon the explicit aging indicators (e.g. CPU utilization) and a state-of-the-art approach based upon implicit aging indicator i.e. $H\ddot{o}lder$ exponent, 
CHAOS increases the detection  accuracy by about 5 times and reduces the $ATTF$ significantly even by 3 orders of magnitude. According to our best knowledge, this is the first work to leverage  entropy  to indicator software aging. The contribution of this paper is three-fold: 
\begin{itemize}
\item We demonstrate that entropy increases with software aging and verify this conclusion via experimental practice and quantitative proof.
\item We propose a novel aging indicator named MMSE. MMSE employs the complexity embedded in  multiple runtime performance metrics to measure software aging and leverages multi-scale and multi-dimension integration to tolerate system fluctuations ,which makes MMSE satisfy the properties: \textit{Stability}, \textit{Monotonicity} and \textit{Integration}.  
\item We design and implement a proof-of-concept named CHAOS, and evaluate the accuracy of three  failure detection approaches based upon  MMSE encapsulated in CHAOS in a VoD system and a real production system, AntVision. The experimental results show that CHAOS improves the detection accuracy by about 5 times and reduces the $ATTF$ by 3 orders of magnitude compared to previous approaches.
\end{itemize}

The rest of this paper is organized as follows. We demonstrate the motivations of this paper in Section  \uppercase\expandafter{\romannumeral 2}. Section \uppercase\expandafter{\romannumeral 3} shows our solution to detect the failure-prone state and the overview of CHAOS.  And in Section \uppercase\expandafter{\romannumeral 4}, we describe the detailed design of CHAOS including: metric selection, MSE and MMSE calculation procedure, and failure-prone state detection approaches. Section \uppercase\expandafter{\romannumeral 5} shows the evaluation results and comparisons to previous approaches. In Section \uppercase\expandafter{\romannumeral 6} we state the related work briefly. Section \uppercase\expandafter{\romannumeral 7} concludes this paper.


%\IEEEPARstart{T}{his} demo file is intended to serve as a ``starter file''
%for IEEE Computer Society journal papers produced under \LaTeX\ using
%IEEEtran.cls version 1.8a and later.
%I wish you the best of success.

%\hfill September 17, 2014
\section{Motivation}
The accuracy of  Aging-Oriented Failure (AOF) detection approaches is largely determined by the aging indicators. A well-designed aging indicator can precisely indicate the AOF. If the subsequent rejuvenations are always conducted  at the real failure-prone state, the rejuvenation cost will tend to be optimal. But unfortunately, prior detection approaches based upon explicit aging indicators \cite{0,1,3,4,14,19,20} don't function well especially in the face of dynamic workloads. They either miss some failures leading to a low recall or mistake some normal states as the failure states leading to a low precision. The insufficiency of previous indicators motivates us to seek novel indicators. We describe our motivations from the following aspects.  
\subsection{Insufficiency of Explicit Aging Indicators}
To distinguish the normal state and failure-prone state, a threshold should be preset on the aging indicator. Once the aging indicator exceeds the threshold, a failure occurs. Traditionally, a threshold is set on explicit aging indicators. For instance, if the CPU utilization exceeds $90\%$,  a failure occurs. However, it's not always the case.  The external observations do not always reveal  accurately the internal states. Here the internal states can be referred to as some normal events (e.g. a file reading, a packet sending) or abnormal events (e.g. a file open exception, a round-off error) generated in the system. In this paper we are more concerned about the abnormal events. Commonly, the internal state space is much smaller than the directly observed external state space. For example, the observed CPU utilization can be any real number in the range $0\% \sim 100\%$ while the abnormal events are very limited. Therefore an abnormal event may correlate with multiple observations. Still take the CPU utilization for example. When a failure-prone event happens, the CPU utilization may be  99\%,  80\% or even 10\%. Therefore the explicit aging indicator can not signify AOF sufficiently and accurately. And if the system fluctuation is taken into account, the situation may get even worse. And this is also a reason why it's so difficult to set an optimal threshold on the explicit aging indicators in order to obtain an accurate failure detection result.
\subsection{Entropy Increase in VoD System}
As the explicit aging indicators fall short in detecting Aging-Oriented Failure, we turn to implicit aging indicators for help. Some insights can be attained from \cite{24} and \cite{14}. Both of them treated software aging  as a complex process. Motivated by them, we believe entropy as a measurement of complexity has a potential to be an implicit aging indicator. 
%And according to our analysis on the relationships between software aging and biology aging in \cite{25} , we are more convinced that some methodologies developed in biological area could be applicable in software aging field. Thus, MSE which is first adopted in biological studies \cite{26,27} is introduced. 

In a real campus VoD (Video on Demand) system which is charge of sharing movies amongst students, we observe that entropy increases with the degree of software aging. The VoD system runs for 52 days until a failure occurs. By manually investigating the reason of failure, we assure it is an Aging-Oriented Failure. During the system running, the CPU utilization  is recorded to be processed later shown in Figure 1. We adopt MSE to calculate the entropy value of the CPU utilization of each day. The result is demonstrated in Figure 2. Figure 2 only shows the entropy value of the first four days ($Day1, Day2, Day3, Day4$) and the last four days ($Day49, Day50, Day51, Day52$). It's apparent to see the entropy values of the last four days are much larger than the ones of the first four days nearly at all scales. Especially, the entropy value of Day52 when the system failed is different significantly from others. However the raw CPU utilization at failure state seems normal which means we may not detect the failure state if using this metric as an aging indicator. Therefore, MSE seems a potential  aging indicator in this practice. %Based upon MSE, an AOF can be detected accurately even by a simple threshold-based method.  
\begin{figure}[!t]
\centering
\includegraphics[width=3.4in,height=2.5in]{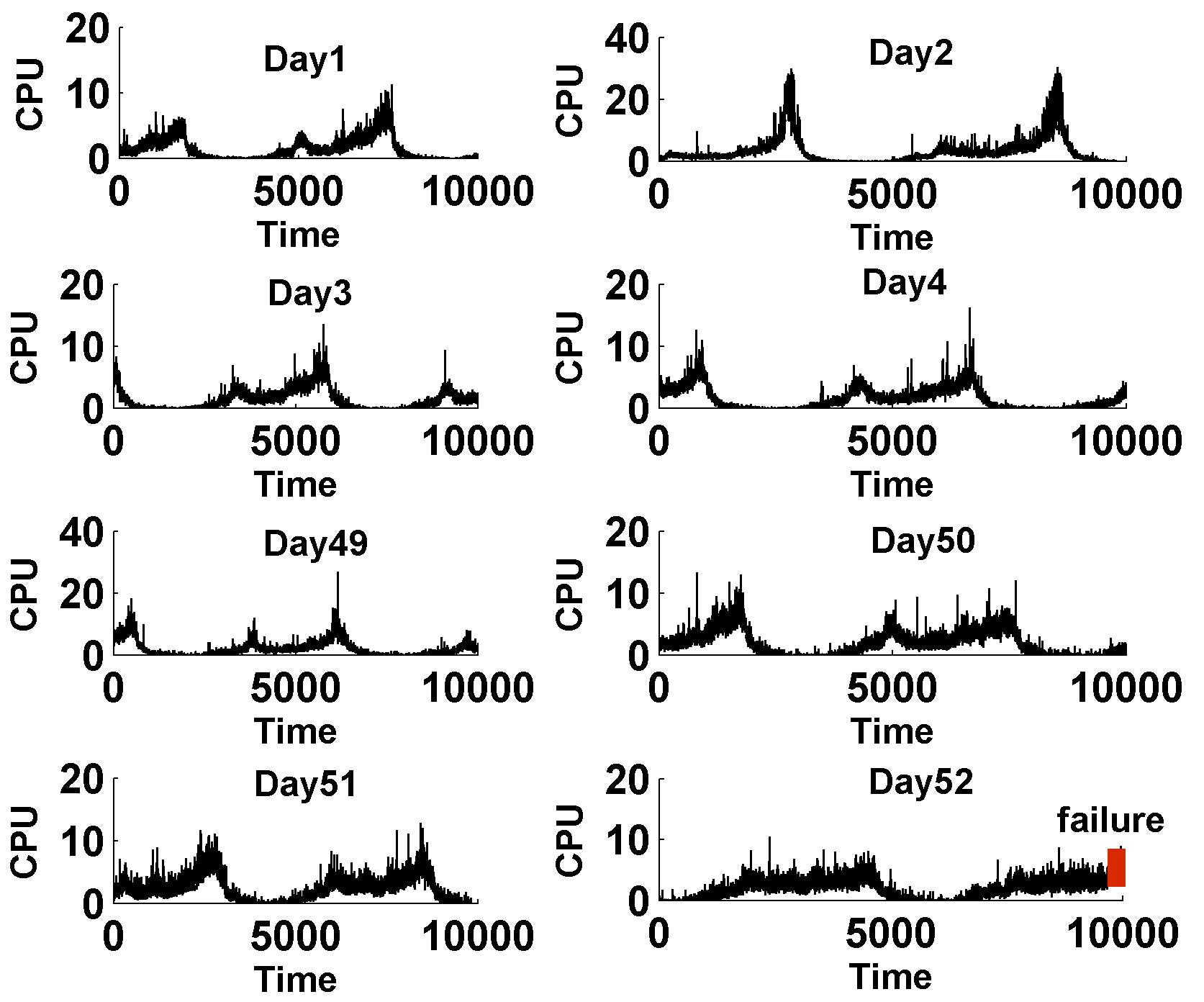}
% where an .eps filename suffix will be assumed under latex, 
% and a .pdf suffix will be assumed for pdflatex; or what has been declared
% via \DeclareGraphicsExtensions.
\caption{The CPU utilization of a real VoD system. In this figure, we only show the CPU utilization of the first four days and the last four days.}
\label{Realdata}
\end{figure}
\begin{figure}[!t]
\centering
\includegraphics[width=3.3in,height=2.0in]{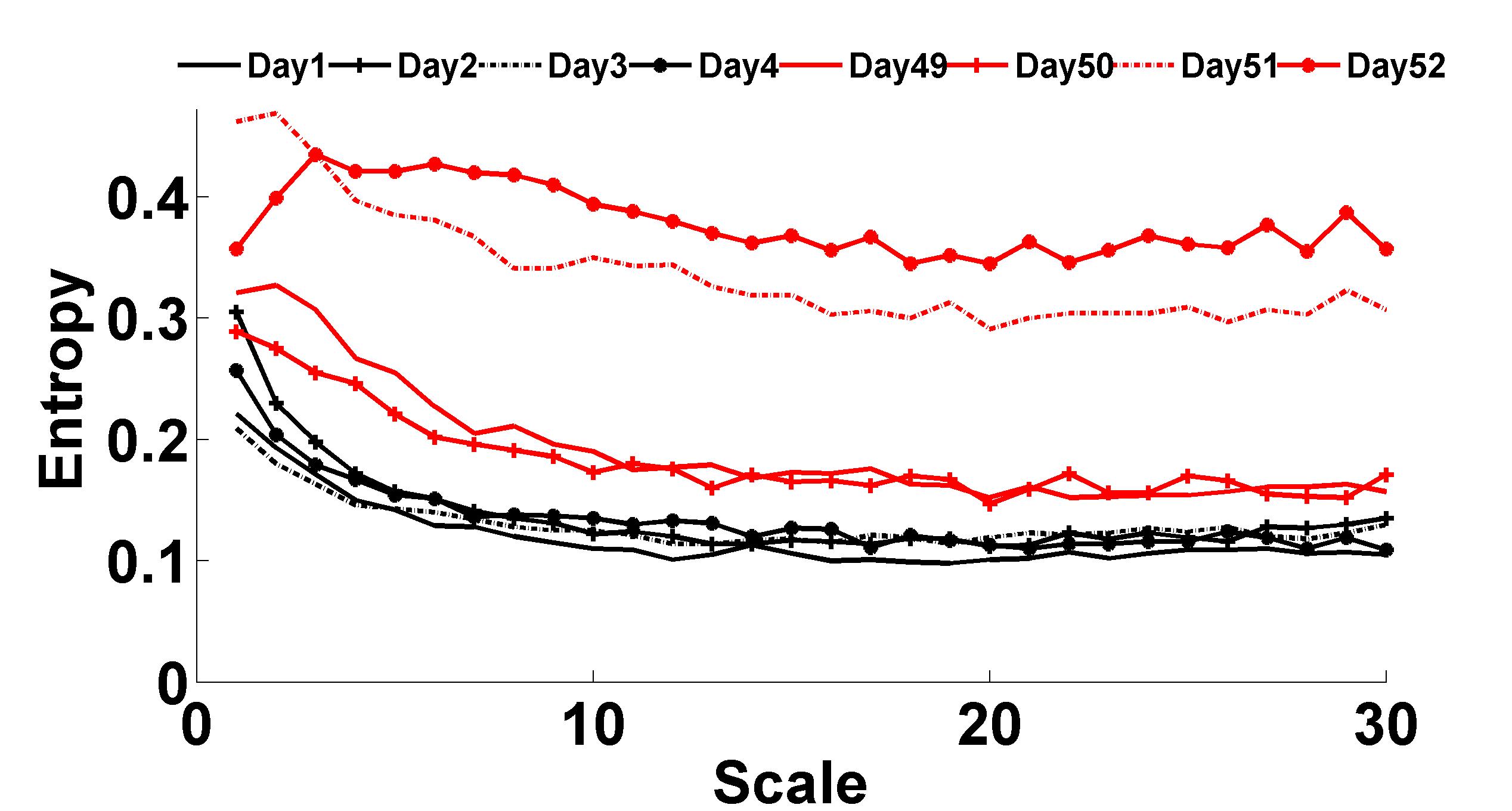}
% where an .eps filename suffix will be assumed under latex, 
% and a .pdf suffix will be assumed for pdflatex; or what has been declared
% via \DeclareGraphicsExtensions.
\caption{The entropy value of a real VoD system at 30 scales}
\label{entropyreal}
\end{figure}
\subsection{Conjecture}
According to the above observation, we provide a high level abstraction of the properties that an ideal aging indicator should satisfy. 
\textit{\textbf{Monotonicity:}} Since software aging is a gradual deterioration  process, the aging indicator should also change consistently with the degree of software aging, namely increase or decrease monotonically. As the most essential property, monotonicity provides a foundation to detect Aging-Oriented Failure accurately.
\textit{\textbf{Stability:}} The indicator is capable of tolerating the noise or disturbance involved in the runtime performance metrics. 
\textit{\textbf{Integration:}} As software aging is a complex process  affected by  multiple factors, the indicator should  cover these influence from multiple data sources, which means it is the integration of multiple runtime metrics. 

It's worth noting that the property set my not be complete, any new property which can strength the detection power of aging indicators can be complemented. In a real-world system, it is extraordinarily hard to find such an ideal aging indicator. But it is possible to find a workaround which is close to the ideal indicator.
\section{Solution}
To provide accurate and effective  approaches to detect AOF, the first step is to propose an appropriate aging indicator satisfying the three properties mentioned in section \uppercase\expandafter{\romannumeral 2}.C. As described in the motivation, we find out MSE seems a potential indicator. But to satisfy all the three properties we proposed, some proofs and modifications are necessary. First of all, we need to quantitatively prove that entropy \footnote{As MSE is a special form of entropy, the properties of entropy are shared by MSE.} caters to $Monotonicity$ in software aging procedure which is illustrated in Appendix A. The proof tells us the system entropy increases with the degree of software aging when the probability of failure state ($p_{f}$) is smaller than the probability of working state ($p_{w}$). In most situations, the system can't provide acceptable services or goes to failure very soon once $p_{w}<p_{f}$.  Therefore we only take into account the scenario with a constraint  $p_{w}>p_{f}$. Under this constraint, the $Monotonicity$ of entropy in software aging is proved. However, the strict monotonicity could be biased a little due to the ever-changing runtime environment. Because of  the inherent ``multi-scale" nature of MSE, the $Stability$ property is strengthened. Via multi-scale transformation, some noises are filtered  or smoothed. In addition, the combination of entropy at multiple scales further mitigates the influence of noises.  The last but not the least property is $Integration$. Unfortunately, MSE is originally designed for analyzing single dimensional data rather than  multiple dimensional data. Thus, to satisfy integration property, we extend the original MSE to MMSE via several modifications. Finally, we achieve a novel software aging indicator, MMSE, which satisfies all the three properties. Based upon MMSE, we have implemented threshold based and time series based methods to detect AOF. To evaluate the effectiveness and accuracy of our approaches, we design and implement a proof-of-concept named CHAOS. The details of  CHAOS will be depicted in next section.
\section{System Design}
The architecture of CHAOS is shown in Figure 3. CHAOS mainly contains four modules: data collection, metric selection, MMSE calculation and crash detection. The data collection module collects runtime performance metrics from multiple data sources including application (e.g. response time), process (e.g. process working set) and operating system (e.g. total memory utilization). Amongst the raw performance metrics, collinearity is thought to be common which means some metrics are redundant. What's worse, a significant overhead is caused if all of performance data is analyzed by the  MMSE calculation module. Thus, a metric selection module is necessary to select a subset of the original metrics without major loss of quality. The selected metric subset is fed into MMSE calculation module to calculate the sample entropy at multiple scales in real time. Then the entropy  values are adopted to detect AOF by the crash detection module. The final result of CHAOS is a boolean value indicating whether failure-prone state occurs. We will demonstrate the details in the following parts.   
\begin{figure}[!t]
\centering
\includegraphics[width=3.3in,height=2.0in]{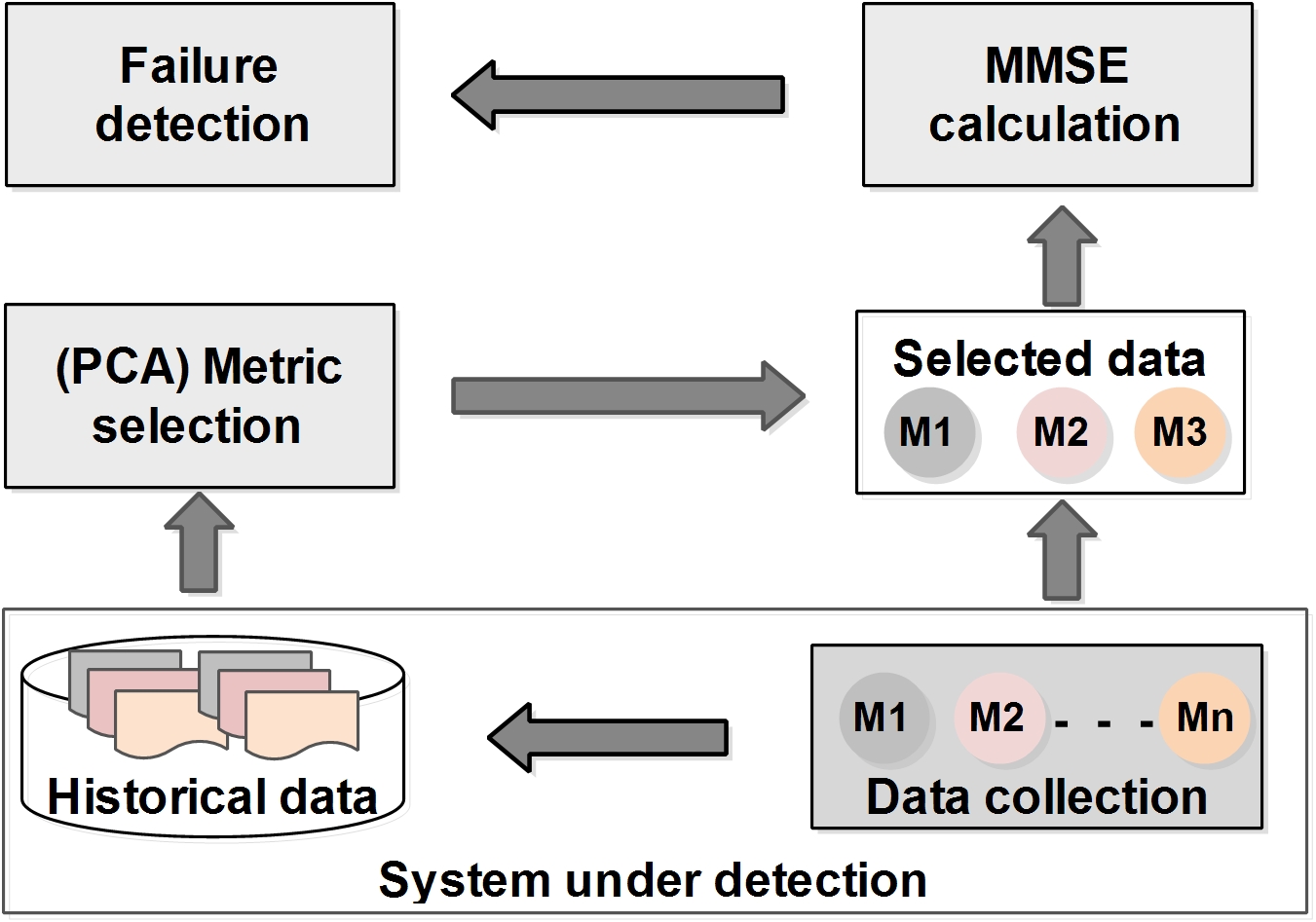}
\caption{The architecture of CHAOS}
\label{Entropy}
\end{figure}
\subsection{Metric Selection}
To get rid of the collinearity amongst the high-dimensional performance metrics and reduce computational overhead, we select a subset of metrics which can be used as a surrogate for a full set of metrics, without significant loss of information. Assume there are $M$ metrics, our goal is to select the best subset of any size $k$ from 1 to $M$.  To this end, PCA (Principal Component Analysis) variable selection method is introduced. 

As a classical  multivariate analysis approach, PCA is always used to transform orthogonally a set of variables which may be correlated to a set of variables which are linearly uncorrelated (i.e. PC).  let $\textbf{X}$ denote a column centered $n$x$M$ matrix, where $M$ denotes the number of metrics, $n$ denotes the number of observations. Via PCA, the matrix $\textbf{X}$ could be reconstructed approximately by $p$ PCs, where  $p \ll M$. These PCs are also called latent factors which are given new physical meanings. Mathematically, $\textbf{X}$ is transformed into a new $n$x$k$ matrix of principal component scores $\textbf{T}$ by a loading or weight $k$x$M$ matrix $\textbf{W}$ if keeping only the $k$ principal component, namely $\textbf{T}=\textbf{X}{{\textbf{W}}^{T}}$ where each column of  $\textbf{T}$ is called a PC. The loading factor $\textbf{W}$ can be obtained by calculating the eigenvector of ${{\textbf{X}}^{T}\textbf{X}}$ or via  singular value decomposition (SVD)\cite{28}. In stead, we leverage PCA to select variables rather than reduce dimensions. 

In order to achieve that goal, we  first introduce a well-defined numerical criteria in order to rank the subset of variables. %Three kinds of criteria , RM \cite{29,30}, GCD (Generalized Coefficient of Determination) \cite{31,30} and RV \cite{32,30}, are suggested in previous literatures. They are all defined as matrix correlations between an $n$x$M$ data matrix and projections of its columns on some subspace of ${R}^{n}$. 
Here choose GCD \cite{30,31} as a criteria. GCD is a measurement of the closeness of two subspaces spanned by different variable sets. In this paper, GCD is a measure of similarity  between the principal subspace spanned by the $k$ specified PCs and the subspace spanned by a given $p$-variable subset of the original $M$-variable data set. By default, the specified PCs are usually the first $k$ PCs and the number of variables and PCs is the same ($k=p$). The detailed description of GCD could be found in \cite{30}. % Mathematically, GCD is expressed in the following equation \cite{30}. 

Then we  need a search algorithm to seek the best $p$-variable subset of the full data set. %Furnival and Wilson \cite{33} developed a simple and fast complete search algorithm named ``leap and bound" for subset selection of a full data set. Base on this algorithm, Duarte Silva \cite{34} proposed another extended algorithm which could be named ``\textit{eleap}". ``\textit{eleap}"  shows great power to select the best subset of a moderate number of variables (e.g. $<$30). While it becomes computationally intractable when the number of variables is large (e.g. $>$35). In our system, we have more than 70 variables to be analyzed. 
In this paper, we adopt a heuristic simulated annealing algorithm to search for the best \textit{p}-variable subset. The algorithm is described in detail in \cite{29}. In brief, an initial \textit{p}-variable subset is fed into the simulated annealing algorithm, then the GCD criterion value is calculated. Further, a subset in the neighborhood \footnote{The neighborhood of a subset \textit{S} is defined as a group of \textit{k}-variable subsets which differ from \textit{S} by only a single variable.} of the current subset is randomly selected. The alternative subset is chosen if its GCD criteria value is larger than the one of the current subset or with a probability $e^{\frac{ac-cc}{t}}$ if the GCD criteria value of the alternative subset (ac) is smaller than the one of current subset (cc) where $t$ denotes the temperature and decreases throughout the iterations of the algorithm. The algorithm stops when the number of iterations exceeds the preset threshold. The merit of the simulated annealing algorithm is that the best \textit{p}-variable subset can be obtained with a reasonable computation overhead even the number of variables is very large. 

With  the well-defined GCD criteria and the simulated annealing search algorithm, we can reduce the high-dimensional runtime performance metrics (e.g. 76) to very low-dimensional data set  (e.g. 5) with very little information loss. And the computation overhead is decreased significantly.  

\subsection{Proposed Multidimensional Multi-scale Entropy}
A well-known measurement of system complexity is the classical %information entropy or 
Shannon entropy \cite{35}. However, Shannon entropy is only concerned with  the instant entropy at a specific time point. It can't capture the temporal structures of one time series completely leading to statistical characteristic loss and even false judgment. MSE proposed by Costa et al \cite{27} is used to quantify the amount of structures (i.e. complexity) embedded in the time series at multiple time scales. A system without structures would exhibit a significant entropy decrease with an increasing time scale. The algorithm of MSE includes two phases:sample entropy \cite{59} calculation and coarse-graining. %Sample entropy \cite{59} which stems from approximate entropy \cite{59} is introduced to quantitatively and numerically measure the regularity or complexity of a time series with limited length. 
Given a positive number $m$, a random variable $X$ and a time series $\textbf{X}=\{X(1),X(2), \cdots, X(N)\}$ with length $N$, $\textbf{X}$ is partitioned into consecutive segments. Each segment is represented by a $m$-length vector: $u_{m}(t)=\{x(t),X(t+1), \cdots, X(t+m-1)\}, 1\leq t \leq N-m+1$ where $m$ could be recognized as the embedded dimension and recommended as $m=2$ \cite{62}. let $n_{i}^{m}(r)$ denote the number of segments that satisfy $d(u_{m}(i),u_{m}(j)) \leq r, i \neq j$ where $i \neq j$ guarantees that self-matches are excluded, %i.e. segments are not compared to themselves,  
$r$ is a preset threshold indicating the tolerance level for two segments to be considered similar and recommended as $r=1.5*\sigma$ \cite{62} where $\sigma$ is the standard deviation of the original time series.  $d(\cdot)=max \{|X(i+k)-X(j+k)|: 1 \leq k \leq m-1 \}$  represents the maximum of the absolute values of differences between $u_{m}(i),u_{m}(j)$ measured by Euclidean distance which is adopted in this paper
%or Chebyshev distance.  
Let $lnC_{i}^{m}(r)=ln\frac{u_{m}(i)}{N-m}$ represent the natural logarithm of the probability that any segment $u_{m}(j)$ is close to segment $u_{m}(i)$, the average of $lnC^{m}(r)$ is expressed as: 
\begin{equation}
\Phi^{m}(r)=\frac{\sum_{i}^{N-m+1}lnC_{i}^{m}(r)}{N-m+1}
\end{equation}
The sample entropy is formalized as:
\begin{equation}
S_{E}(m,r,N)=-ln\frac{\Phi^{m+1}(r)}{\Phi^{m}(r)}
\end{equation} 
To ensure $\Phi^{m+1}(r)$ is defined in any particular $N$-length time series, sample entropy redefines $\Phi^{m}(r)$ as:
\begin{equation}
\Phi^{m}(r)=\frac{\sum_{i}^{N-m}lnC_{i}^{m}(r)}{N-m}
\end{equation}

Suppose  $\tau$ is the scale factor, the consecutive coarse-grained time series $Y^{\tau}$ is constructed in the following two steps:
%1.Divide the original time series $\{X(1),X(2), \cdots, X(N)\}$  into consecutive and non-overlapping windows of length $\tau$;\\
%2.Average the data points inside each window;
\begin{itemize}
\item Divide the original time series $\textbf{X}$  into consecutive and non-overlapping windows of length $\tau$;
\item Average the data points inside each window;
\end{itemize} 
Finally we get $Y^{\tau}=\{y_{j}^{\tau}|: 1 \leq j \leq  \lfloor \frac{N}{\tau} \rfloor  \}$ and each element of  $Y^{\tau}$ is defined as:
\begin{equation}
y_{j}^{\tau}=\frac{ \sum_{i=(j-1)\tau+1}^{j\tau}X(i)}{\tau},  1 \leq j \leq  \lfloor \frac{N}{\tau} \rfloor
\end{equation} 
When $\tau=1$, $Y^{\tau}$ degenerates to the original time series $\textbf{X}$. Then MSE of the original time series $\textbf{X}$ is obtained by computing the sample entropy of  $Y^{\tau}$ at all scales.  
However, the conventional MSE is designed for single dimensional analysis. Thus, it doesn't satisfy the property \textit{Integration} of an aging indicator. To this end, we extend MSE to MMSE via several modifications. 
    
\textbf{Modification 1.}  The collected multi-dimensional performance metrics usually have different scales and numerical ranges. For example the CPU utilization  metric stays in the range of 0 $\sim$ 100 percentage while the total memory utilization may vary in the range 1048576KB $\sim$ 4194304 KB. Thus, the distance between two segments may be biased by the performance metrics with large numerical ranges, which further results in MSE bias. To avoid that bias, we normalize all the performance metrics to a unified numerical range,namely $0 \sim 1$. Suppose $X$ is a $Nxp$ data matrix where $p$ is the number of performance metrics, $N$ is the length of the data window and  each column of $X$ denotes the time series of one particular performance metric, then $X$ is normalized in the following way:
\begin{equation}
X_{ji}^{'}=\frac{X_{ji}-min(X_{i})}{max(X_{i})-min(X_{i})}, 1 \leq i \leq p, 1 \leq j \leq N 
\end{equation}   

\textbf{Modification 2.} In MSE algorithm, we quantify the similarity between two segments via  maximum norm \cite{37} of two scalar numbers. A novel quantification approach is necessary when  MSE is extended to MMSE. Each element in the maximal norm pair: $max\{|X(i+k)-X(j+k)|: 1 \leq k \leq m-1 \}$ such as $X(i+k)$ is replaced by a vector $\textbf{X}(i+k)$ where each element represents the observation of  one specific performance metric at time $i+k$. Thus the scalar norm is transformed to the vector norm. The embedded dimension $m$ should also be vectorized when the analysis shifts from single dimension to multiple dimensions. The vectorization brings a nontrivial problem in the calculation procedure of sample entropy that is how to obtain $\phi^{m+1}(r)$.  Assume that the  embedding vector $\textbf{m}=(m_{1},m_{2}, \cdots,m_{p})$ denotes the embedded dimensions for $p$ performance  metrics respectively. A new  embedding vector $\textbf{m}^{+}$ which has one additional dimension compared to $\textbf{m}$ can be obtained in two ways. The first approach comes from the study in \cite{37}. According to the embedding theory mentioned in \cite{38}, $\textbf{m}^{+}$ can be achieved by adding one additional dimension to only one  specific embedded  dimension in $\textbf{m}$, which leads to $p$ different alternatives. $\textbf{m}^{+}$ can be any one of the set  $\{ (m_{1},m_{2},\cdots, m_{k}+1, \cdots,  m_{p}), 1 \leq k \leq p\}$. $\phi^{\textbf{m}^{+}}(r)$ is calculated in a naive way or a rigorous way both of which are depicted in detail in \cite{37}.  The other approach is very simple and intuitional that is adding one additional dimension to every embedded  dimension in $\textbf{m}$. There is only one alternative for $\textbf{m}^{+}$ namely $\{ (m_{1}+1,m_{2}+1,\cdots, m_{k}+1, \cdots,  m_{p}+1), 1 \leq k \leq p\}$ . This simple approach implies that each embedded dimension is identical, which may be a strong constraint. However, compared to the former approach, the latter one has negligible computation overhead and works well in this paper.  The former approach will be discussed in our future work.   

\textbf{Modification 3.} In  MSE algorithm, the threshold $r$ is set as $r=0.15*\sigma$.% to determine whether two segments are similar where  $\sigma$ is the standard deviation of the original time series.  
%An extension for the stand deviation  $\sigma$ of multivariate in MMSE is a necessity. 
 In MMSE algorithm, we need a single number to represent the variance of  the  multi-dimensional  performance data in order to apply it directly in the similarity calculation procedure. Here we employ the total variance denoted by $tr(\textbf{S})$  which is defined as the trace of the covariance $\textbf{S}$ of the normalized multi-dimensional performance data to replace $\sigma$.         

\textbf{Modification 4.} We argue that an ideal aging indicator should be expressed as a single number in order to be readily used in failure detection. The output of the conventional MSE is a vector of entropy values at multiple scales. We need to use a holistic metric to integrate all the entropy values at multiple scales. Thus a composed entropy (CE) is proposed. Let $T$ denote the number of scales and the vector $\textbf{E}=(e_{1}, e_{2}, \cdots, e_{T})$ denote the entropy value at each scale respectively. Then  $CE$ is defined as the Euclidean norm of the entropy vector $\textbf{E}$ :
\begin{equation}
CE=\sqrt[2]{ \sum_{i=1}^{T}e_{i}^2 }
\end{equation}
$CE$ cloud be regarded  as the Euclidean distance between $\textbf{E}$ and a ``zero" entropy vector  which consists of  $0$ entropy values. A ``zero" entropy vector represents an ideal system state meaning that the system runs in a health state without any fluctuations.  Thus the more $\textbf{E}$ deviates from a ``zero" entropy vector, the worse the system performance is. It's worth noting that $CE$ is not the unique metric which can integrate the entropy values at all scales. Other metrics also have the  potential to be the aging indicators. For example, the average of $\textbf{E}$ is another alternative although we observe that it has a consistent result with $CE$.

Through the aforementioned modifications on MSE,  the novel aging indicator MMSE  has satisfied all the three properties: \textit{Monotonicity}, \textit{Stability} and \textit{Integration} proposed in Section \uppercase\expandafter{\romannumeral 2}.C. %Therefore, MMSE can be leveraged as aging indicator to detect  failure-prone state caused by software aging. 
For the sake of clarity, we demonstrate the pseudo code of MMSE algorithm  in Algorithm 1. 
\begin{algorithm}
\caption{MMSE algorithm}
\label{alg1}
\begin{algorithmic}[1]
\REQUIRE $m$:the embedded dimension; $T$:the number of scales; $N$:the length of data window; $\textbf{X}$: a $N$x$p$ data matrix where each $p$ denotes the number of performance metrics and each column $X_{i}, 1 \leq i \leq p$ denotes the time series of one specific performance metric with length $N$.
\ENSURE The aging degree metric $CE$
\STATE $//$ Normalize the original time series into the range [0,1]
\FOR {$j=1;j=N;j++$}
\FOR {$i=1;i=p;i++$}
\STATE  $X_{ji}^{'}=\frac{X_{ji}-min(X_{i})}{max(X_{i})-min(X_{i})}$
\ENDFOR 
\ENDFOR 
\STATE $//$ Preset the similarity threshold $r$
\STATE $S=Cov(X^{'})$  $//$ $Cov$ denotes the matrix covariance 
\STATE $r=tr(S)$ $//$ $tr$ denotes the trace  of a particular matrix 
\FOR {$\tau=1;\tau=T;\tau++$}
\STATE $//$ Coarse-graining procedure
\FOR {$i=1;i=p;i++$}
\FOR {$j=1;j=\lfloor \frac{N}{\tau} \rfloor;j++$}
\STATE $Y_{ji}=\frac{ \sum_{k=(j-1)\tau+1}^{j\tau}X_{ki}^{'}}{\tau}$
\ENDFOR
\ENDFOR
\STATE $E(\tau)=ExtendedSampleEntropy(m,r,Y)$ 
\STATE $//$ The similarity calculation between two 
\STATE $//$ segments has been extended from scalar 
\STATE $//$ to vector in $ExtendedSampleEntropy(\cdot)$
\ENDFOR 
\STATE $//$ Calculate the composed entropy $CE$
\STATE $CE=\sqrt[2]{ \sum_{i=1}^{T}E(i)^2 }$
\end{algorithmic}
\end{algorithm}
\subsection{AOF Detection based upon MMSE}
Based upon the proposed aging indicator MMSE, it's easy to design algorithms to detect AOF in real time. According to the survey \cite{18}, there are three kinds of approaches including \textit{time series analysis,threshold-based} and \textit{machine learning} to detect or predict the occurrence of AOF.% by analyzing the collected runtime data. 
In this paper, we only discuss the time series and threshold-based approaches and leave the machine learning approach in our future work. But before that we need to determine a sliding data window in order to calculate MMSE  in real time.  As mentioned in previous work \cite{61},  $\lfloor \frac{N}{\tau} \rfloor$ should stay in the range $10^{m}$ to $30^{m}$. %and $m=2$. 
Thus the sliding window heavily depends on the scale factor $\tau$. In previous studies \cite{26,27,37}, they  usually set the scale factor $\tau$ in the range $1\sim20$ leading to a huge data window, say 10000, especially when $\tau=20$. A large sliding window not only increases the computational overhead but also makes detection approaches  insensitive to failure. Thus we constrain the sliding window in an appropriate range, say no more than 1000, by limiting the range of $\tau$. In this paper we set $\tau$ in the range $1\sim10$. So a moderate data window $N=1000$ can cater the basic requirement.\\
\textbf{Threshold based approach.} As a simple and straightforward approach, the threshold based approach is widely used in aging failure detection \cite{22,23}. If the aging indicator exceeds the preset threshold, a failure occurs. However an essential challenge is how to identify an appropriate threshold. %which is also discussed in the motivation section in detail. 
Identifying the threshold from the empirical observation is a feasible approach. %According to the type of the observed data, we design two approaches. The first approach directly determines the failure threshold by extracting the data from failure history, which is called $FailureThreshold$ approach. Assume that $\textbf{CE}=\{CE_{1},CE_{2},CE_{3},\cdots,CE_{n}\}$ represents a group of failure data where each element $CE_{i}$ denotes the $CE$ value when the $i$th Aging-Oriented-Failure occurs. The failure threshold can be set as the average of these $CE$s, namely $\frac{\sum_{i=1}^{n}CE_{i}}{n}$. This approach can explicitly reflect the characteristic of aging failures. But it requires a relatively large failure sample. So when the failure sample is modest, we can leverage another approach.
This  approach learns  a normal pattern %from  historical runtime data 
when the system runs in the normal state. If the normal pattern is violated, a failure occurs. We call this approach  $FailureThreshold$ ($FT$). Assume that $\textbf{CE}=\{CE(1),CE(2),CE(3),\cdots,CE(n)\}$ represents a series of normal data where each element $CE(t)$ denotes a $CE$ value at time $t$. The failure threshold $ft$ is defined as: $ft=\beta*max(\textbf{CE})$
%\begin{equation}
%ft=\beta*max(\textbf{CE})
%\end{equation}
where $\beta$ is a tunable fluctuation factor which is used to cover the unobserved value escaped from the training data. % If a new observed $CE$ exceeds $ft$, a failure occurs. 
As mentioned above, MMSE increases with the degree of software aging. Thus  a failure occurs only when the new observed $CE$ exceeds $ft$, something like upper boundary test. For the aging indicators which have a downtrend such as \textit{AverageBandwidth}, the $max$ function in (9) will be replaced by $min$, something like lower boundary test. A failure occurs if the new observed $CE$ is lower than $ft$. %We will leverage it in the following comparisons.           %We set $\beta=1.2$  in this paper. 

$FT$ can be further extended to be an incremental version named $FT$-$X$ in order to adapt to the ever changing running environment.  $FT$-$X$ learns  $ft$ incrementally from historical data. Once a new $CE(t+1)$ is obtained and the system is assured to stay in the normal state, then we compare $CE(t+1)$ with previously trained $max(\textbf{CE}(t))$. If  $CE(t+1) < max(\textbf{CE}(t))$ then $ft=\beta*max(\textbf{CE})$ else   $ft=\beta*CE(t+1)$. Besides the realtime advantage, $FT$-$X$ needs very little memory space to store the new $CE$ and  previously trained maximum of $\textbf{CE}$. \\ 
\textbf{Time series approach.} Although the threshold based approach is simple and straightforward, identifying the threshold is still a  thorny problem. %What's more, the static threshold can't well adapt to the ever changing runtime environment.  
Thus, to bypass the threshold setting dilemma, we need a time series approach which requires no threshold or adjusts a threshold dynamically. To compare with existing approaches, we leverage the extended version of \textit{Shewhart control charts} algorithm proposed in \cite{24} to detect AOF. But one difference exists. %The deviation $d_{n}$  between the local average $a_{n}$ and the global mean  $\mu_{n}$  is defined as:
In \cite{24}, they adopt the deviation $d_{n}$ between the local average $a_{n}$ and the global mean  $\mu_{n}$ to detect aging failures. $d_{n}$ is defined as:
\begin{equation}
d_{n}=\frac{ \sqrt[2]{N^{'}}}{\sigma_{n}}(\mu_{n}-a_{n})
\end{equation}
where $N^{'}$ is used to represent the sliding window on entropy data calculated by MMSE algorithm in order to distinguish it from the sliding window $N$ in MMSE algorithm,  the meaning of other relevant parameters can be found in \cite{24}. They pointed out that $H \ddot{o}lder$ exponent decreased with the degree of  software aging. Therefore they only took into account the scenario of $\mu_{n}>a_{n}$.  In this paper, we prove that MMSE increases with the degree of software aging. Thus we only take into account the scenario of $\mu_{n}<a_{n}$. $d_{n}$ is redefined as:
\begin{equation}
d_{n}=\frac{ \sqrt[2]{N^{'}}}{\sigma_{n}}(a_{n}-\mu_{n})
\end{equation}
%rather than :
%\begin{equation}
%d_{n}=\frac{ \sqrt[2]{N^{'}}}{\sigma_{n}}(\mu_{n}-a_{n})
%\end{equation}
%
If $d_{n} > \epsilon$ holds for $p$ consecutive points where $\epsilon$ and $p$ are tunable parameters, a change occurs. We insist that a change is assured when $p=4$ at least in this paper. So $N^{'}$ and $\epsilon$ are the primary factors affecting the detection results.  In \cite{24}, the second change in $H \ddot{o}lder$ exponent implies a system failure. By observing the MMSE variation curves obtained from Helix Server test platform and real-world AntVision system shown in Section \uppercase\expandafter{\romannumeral 6}, we find out that these curves can be roughly divided into three phases: slowly rising phase, fast rising phase and  failure-prone  phase. And when the system steps into the failure-prone  phase,  a failure will come soon. Therefore we also assume that the second change in MMSE data implies a system failure. %Other time-series approaches such as CUSUM \cite{9,39} can also be applied in this paper. But due to the limited space, we leave them to our future work.

\section{Experimental Evaluation}
We have designed and implemented a proof-of-concept named CHAOS and deployed it a controlled environment. To monitor the common process and operating system related performance metrics such as CPU utilization and context switch, we employ some off-the-shelf tools such as Windows Performance Monitor shipped with Window OS or Hyperic \cite{46}; to monitor other application related metrics such as response time and throughput, we develop several probes from scratch %and deploy them in the test environment. 
The sampling interval in all the monitoring tools is 1 minute. Next, we will demonstrate the details of our experimental methodology and evaluation results in a VoD system, Helix Server and in a real production system, AntVision.
\subsection{Evaluation Methodology}
To make comprehensive evaluations and comparisons from multiple angels, we deploy CHAOS in a VoD test environment. And to evaluate the effectiveness of CHAOS in real world systems, we use CHAOS to detect failures in  AntVision system.

\textbf{VoD system.}
We choose VoD system as our test platform because more and more services involve video and audio data transmission. What's more, the ``aging" phenomenon has been observed in such kinds of applications in our previous work \cite{40,41}.  We leverage  Helix Server \cite{47} as a test platform to evaluate our system due to its open source and wide usage. Helix Server as a mainstream VoD software system is adopted to  transmit video and audio data via RTSP/ RTP protocol. At present, there are very few VoD benchmarks. Hence, % we have no off-the-shelf tools to construct our test platform. To that end, this paper 
we develop a client emulator named $HelixClientEmulator$ employing RTSP and RTP protocols from scratch. %$HelixClientEmulator$  involves three threads, one for audio processing (e.g. decoding), one for video processing and the third for session management (e.g. start and stop a file request). 
It can generate multiple concurrent clients to access media files on a Helix Server. Our test platform  consists of one server hosting Helix Server, three clients hosting $HelixClientEmulator$ and one Gigabit switch connecting the clients and the server together. 100 rmvb media files with different bit rates are deployed  on the Helix Server machine. %The topology of our test system is demonstrated in Figure 6.  
Each client machine is configured with one Intel dual core 2.66Ghz CPU and 2 GB memory and one Gigabit NIC and runs  64-bit Windows 7 operating system.  The server machine is configured with two 4-core Xeon 2.1 GHZ CPU processors, 16GB memory, a 1TB hard disk and a Gigabit NIC and runs  64-bit Windows server 2003 operating system.

During system running, thousands of performance counters can be monitored. In order to trade off between monitoring effort and information completeness, this paper only monitors some of the parameters at four different levels: Helix Client, OS, Helix Server, and server process via respective probes shown in Figure 6. From Helix Client level, we record the performance metrics such as \textit{Jitter},   \textit{Average Response Time} and etc via the probes embedded in  $HelixClientEmulator$; from OS level, we monitor \textit{Network Transmission Rate},  \textit{Total CPU Utilization} and etc via Windows Performance Monitor; from Helix Server level, we monitor the application relevant metrics such as \textit{Average Bandwidth Output Per Player(bps)}, \textit{Players Connected} and etc from the log produced by Helix Server; from process level, we monitor some of metrics related to the Helix Server process like \textit{Process Working Set} via Windows Performance Monitor.Due to the limited space, we will not show the 76  performance metrics. %In all, 76 metrics are continuously monitored in our test system. 

\textbf{AntVision System.} 
Besides the evaluations in a controlled environment,  we further apply CHAOS to detect failures in  AntVision system. AntVision is a complex system which is used to monitor and analyze  public opinions and information from social networks like Sina Weibo.  The whole system consists of hundreds of machines in charge of  crawling information, filtering data, storing data and etc. More information about this system can be found in  $www.antvision.net$. With the help of system administrators,  we have obtained a $7$-$day$ runtime log from AntVision. The log data not only contain performance data  %collected in    
but also failure reports. Although the performance data only involve two metrics i.e. CPU and memory utilization, it's enough to evaluate the failure detection power of  CHAOS. According to the failure reports, we observe that one machine crashed in the 6th day without knowing the reason. After manual investigation, we conclude that the outage is  likely caused by software aging. %Thus these performance data will be adopted in the following evaluation. 

In the controlled environment, we conducted 50 experiments. In each experiment, we guarantee the system runs to ``failure". Here ``failure" not only refers to system crashes but also QoS violations.%situations when the system can not satisfy the QoS requirement are also regarded as a kind of ``failure". 
In this paper, we leverage  \textit{Average Bandwidth Output Per Player(bps)} (\textit{AverageBandwidth}) as the QoS metric. Once \textit{AverageBandwidth} is lower than a preset threshold e.g. 30bps for a long period,  a ``failure" occurs because a large number of video and audio frames are lost at that moment. To get the ground truth, we manually label the ``failure" point for each experiment. However due to the interference of noise and ambiguity of manual labeling, the failure detection approaches may report failures around  the labeled  ``failure" point rather than at the precise ``failure" point. Thus we determine that the failure is correctly detected if  the failure report falls in the ``decision window". The decision window with a specific length (e.g. 100 in this paper)is defined as a data window whose right boundary is the labeled ``failure" point. %Figure 7 demonstrates the  ``failure" labeling and decision window setting. In Figure 7, a ``failure"  is labeled at 1700 time slot where \textit{AverageBandwidth} stays under 30bps for a period of time. And the data in the range $1600 \sim 1700$ is recognized as a decision window. 
%\begin{figure}[!t]
%%\setlength{\abovecaptionskip}{-1cm}
%\setlength{\belowcaptionskip}{-1cm}
%\centering
%\includegraphics[width=3.3in]{AverageBandwidth}
%\caption{The failure point labeling and decision window setting}
%\label{degradation}
%\end{figure} 

Four metrics are employed to quantitatively evaluate the effectiveness of CHAOS. They are $Recall$, $Precision$, $F1$-$measure$ and $ATTF$. The former two metrics are defined as:
$$Recall=\frac{N_{tp}}{N_{tp}+N_{fn}}, Precision=\frac{N_{tp}}{N_{tp}+N_{fp}}$$  where $N_{tp}$, $N_{fn}$, and $N_{fp}$ denote the number of true positives, false negatives, and false positives respectively. It's worth noting that  $N_{tp}$, $N_{fn}$, $N_{fp}$ are the aggregated numbers  over  50 experiments respectively. To represent the accuracy in a single value,   $F1$-$measure$ is leveraged and defined as: $$ F1-measure=\frac{2*Recall*Precision}{Recall+Precision}$$  $ATTF$ is defined as the time span between the first failure report and the real failure namely the left boundary of the decision window in this paper. In a real-world system, once a failure is detected the system may be rebooted or offloaded for maintenance. Thus we choose the first failure report as a reference point. If the first failure report falls in the decision window, $ATTF=0$. %As described in our motivation, 
A large $ATTF$ may cause excessive system maintenance leading to availability decrease and operation cost increase. Therefore a lower $ATTF$ is preferred.

\subsection{Performance Metric Selection}
%We conducted 50 experiments. In each experiment, the phenomenon of performance degradation namely ``aging" is observed. For example Figure 9 demonstrates the \textit{AverageBandwidth} and \textit{MemoryUtilization} variation curves along with time slot under workload $(700,0,1)$. Even a cursory glance at these curves, a decreasing trend in \textit{AverageBandwidth} and an increasing trend in \textit{MemoryUtilization} can be observed. To exhibit these trends more clearly, we adopt $Lowess$ [] method to smooth these two time series. From the smoothed curves, we observe that software aging indeed exists in VoD system and the system tends to the ``failure" point where serious video and audio frame losses occur.
%\begin{figure}[!t]
%\centering
%\includegraphics[width=3.3in]{lowess}
%\caption{Performance degradation in Helix Server system under workload  $(700,0,1)$. (a) denotes the degradation of Average Bandwidth and (b) denotes the increase of memory utilization. The red dash line denotes the fitted curve using $Lowess$ }
%\label{lowess}
%\end{figure} 
By investigating  all the performance metrics, we find that many metrics have very similar characteristics like trend meaning these metrics are highly correlated. Therefore we select a small subset of metrics which can be used as a surrogate of the full data set without  significant information loss via PCA variable selection presented in Section \uppercase\expandafter{\romannumeral 5}.A. We calculate the best GCD scores of  different variable sets with specific cardinalities (e.g. $k=3$) by the simulated annealing algorithm.  Figure 4 shows the variation of the best GCD sore along with the number of variables.  From this figure, we observe that the GCD score doesn't increase significantly any more when the number of variables reaches 5. Therefore these 5 variables are already capable of representing the full data set. The 5 variables are \textit{Total CPU Utilization}, \textit{AverageBadwidth}, \textit{Process IO Operations Per Second }, \textit{Process Virtual Bytes Peak}, \textit{Jitter} respectively. In the following experiments, we  will  use them to evaluate CHAOS. 
%\begin{figure}[!t]
%\centering
%\includegraphics[width=3.3in]{VS}
%\caption{The variation of GCD sore along with the number of variables}
%\label{VS}
%\end{figure}
\begin{figure}[!t]
\begin{minipage}[t]{0.5\linewidth}
\centering
\includegraphics[width=1.6in,height=1.2in]{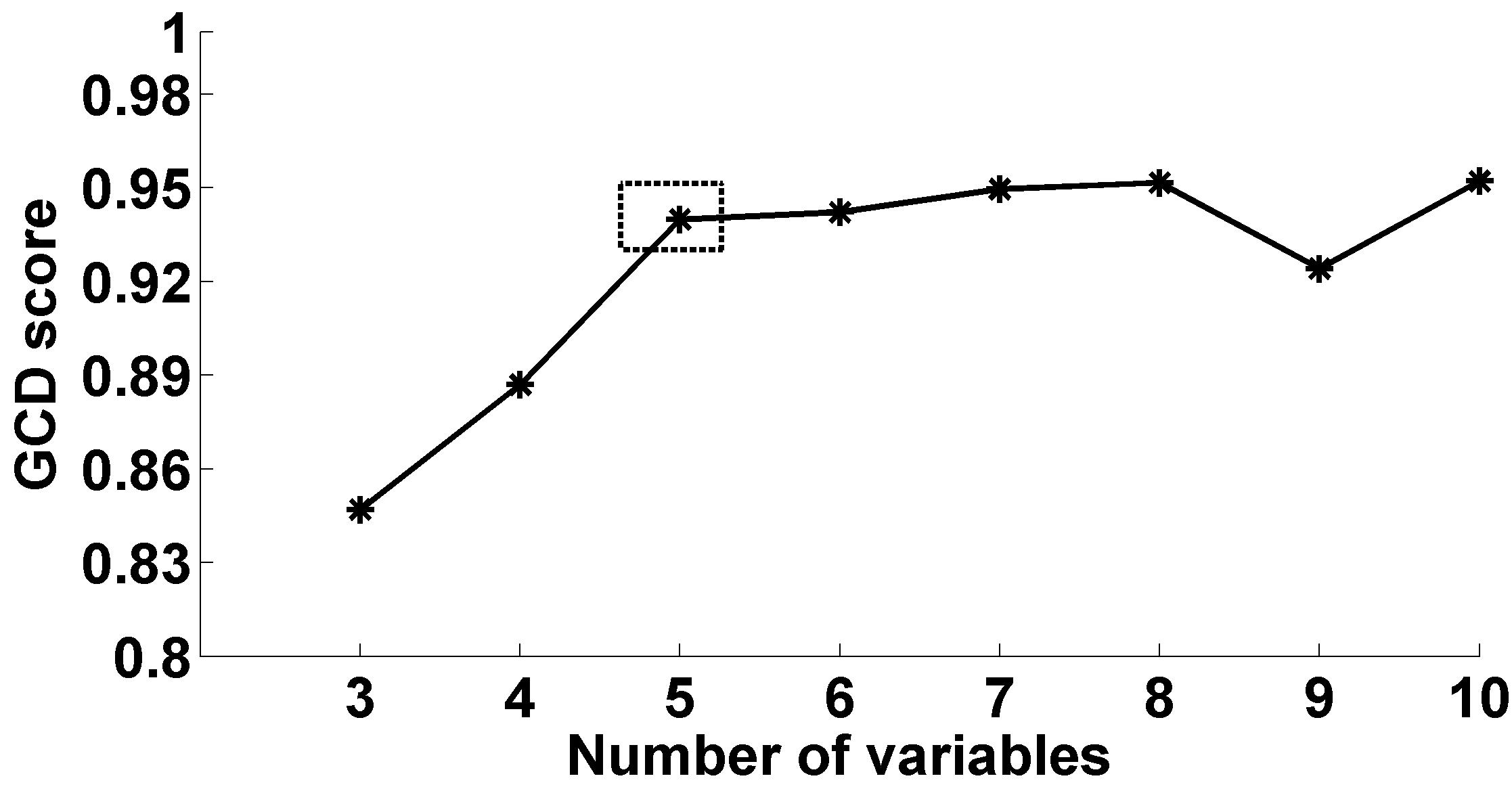}
\caption{The variation of GCD sore along with the number of variables}
\label{Recall}
\end{minipage}%
\begin{minipage}[t]{0.5\linewidth}
\centering
\includegraphics[width=1.6in,height=1.2in]{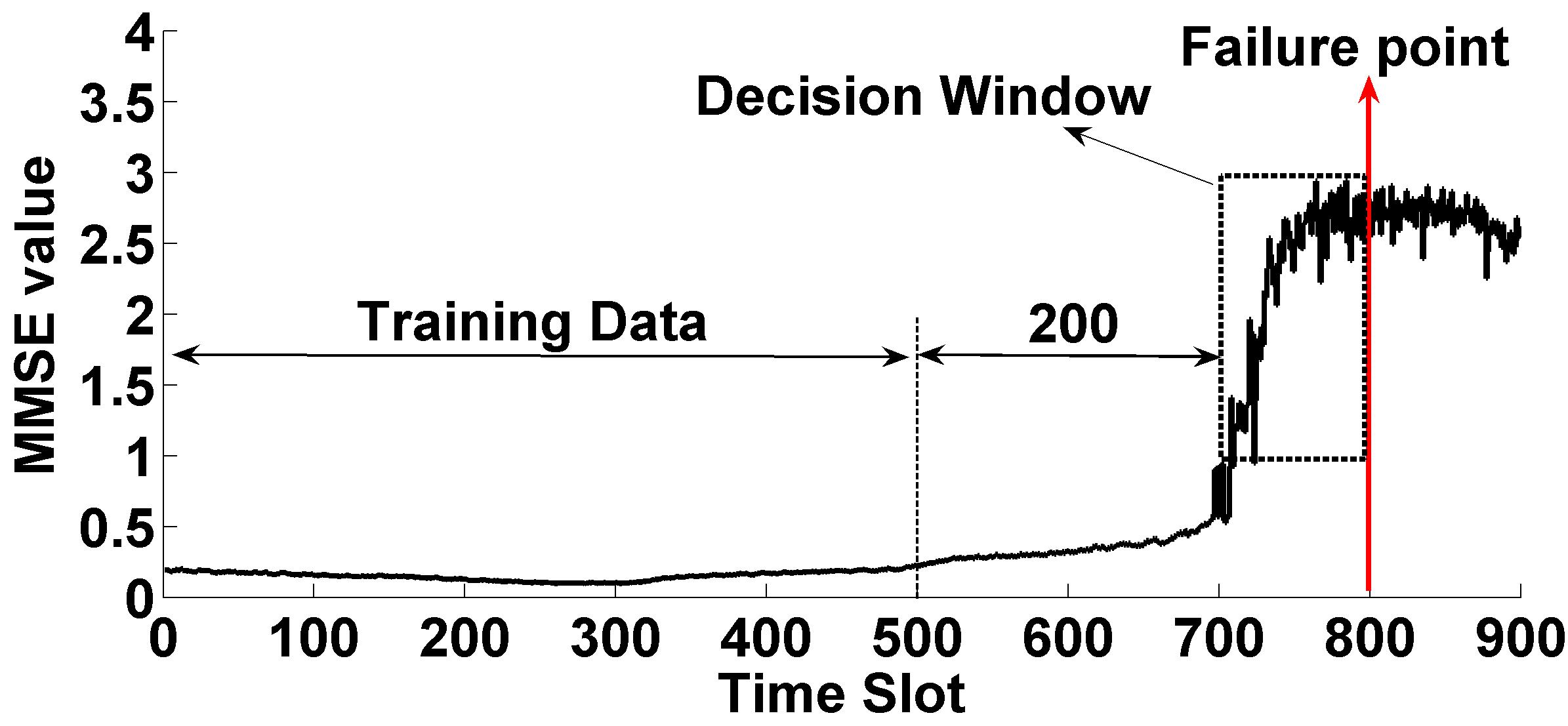}
\caption{Training data selection in $FT$ approach.}
\label{Precision}
\end{minipage}
\end{figure} 
\subsection{AOF Detection}
In this section, we will demonstrate the the failure detection results of CHAOS. In MMSE algorithm, we set the embedded dimension $m=2$, the sliding window $N=1000$,the number of scales $T=10$.  For the failure detection approach $FT$, we need to prepare the  training data and  determine the fluctuation factor $\beta$ first. Due to the lack of prior knowledge, the  training data selection is full of randomness and blindness. To unify the way of training data selection, we leverage the slice of MMSE data ranging from the system start point to the point where 200 time slots away from the right boundary of the decision window as the training data. And leave the left 200 time slots to conduct and compare to $FT$-$X$ approach. Figure 6 shows an example of training data selection in one experiment. In this figure, we set the point in the 800th time slot as the ``failure" point. The decision window spans across the range $700 \sim 800$. Thus the data slice in the range $0 \sim 500$ is selected as the training data.  
%\begin{figure}[!t]
%\centering
%\includegraphics[width=3.3in]{TrainingData700-20}
%\caption{Training data selection in $FT$ approach.}
%\label{TrainingData}
%\end{figure} 

Another problem is how to determine $\beta$.  According to the historical performance metrics and failure records, it's possible to achieve an optimal $\beta$. Figure 6 (a) demonstrates the failure detection results of $FT$ with different $\beta$ values. From this figure, we observe that $Recall$ keeps a perfect value 1 when $\beta$ varies in the range $1\sim2$,  i.e. $N_{fn}=0$ % meaning that there are no false negatives or $N_{fn}=0$ 
and the other two metrics: $Precision$ and $F1$-$measure$ increase with $\beta$. From Figure 5, we can find some clues to explain these observations. In Figure 5, the selected training data in the range  $0 \sim 500$ is  much  smaller than the data in the decision window. Hence, no matter how $\beta$ varies in the range $1 \sim  2$, the failure threshold $ft$ is lower than the data in the decision window. The advantage is that all of the failures can be pinpointed (i.e. $N_{fn}=0$). While the disadvantage is that many normal data are mistaken as failures (i.e. $N_{fp}$ is large). And the $Precision$ has an increasing trend due to the decreasing of $N_{fp}$ with $\beta$. Similarly, the detection results  $FT$-$X$ with different $\beta$ values are shown in Figure 6 (b). But quite different from the  observations in Figure 6 (a), the $Precision$ keeps a perfect value 1 (i.e. $N_{fp}=0$) while the other two metrics $Recall$ and $F1$-$measure$ decrease with $\beta$ in Figure 6 (b). Figure 5 is also capable of  explaining these observations. The failure threshold $ft$ is updated  by $FT$-$X$ incrementally according to the system state. As the system runs normally  in the range $500\sim700$, these data are also used to train $ft$. Hence $max(\textbf{CE})$ calculated by $FT$-$X$ is much bigger than the one calculated by $FT$.  A bigger $\beta$ can guarantee the detected failures are the real failures (i.e. $N_{fp}=0$) but may result in a large  failure missing rate (i.e. $N_{fn}$ is large). From these two figures, we observe that $FT$ achieves an optimal result when $\beta$ is large, say $\beta=2$ but $FT$-$X$ achieves an optimal result when $\beta$ is small, say $\beta=1.1$. To carry out fair comparisons, we set $\beta=2$ for $FT$ and $\beta=1.1$ for $FT$-$X$, namely their optimal results. However in real-world applications, the optimal $\beta$ is considerably difficult to attain especially when failure records are scarce. In that case, $\beta$ can be determined by rule-of-thumb. 
\begin{figure}[!t]
\centering
\includegraphics[width=3.3in]{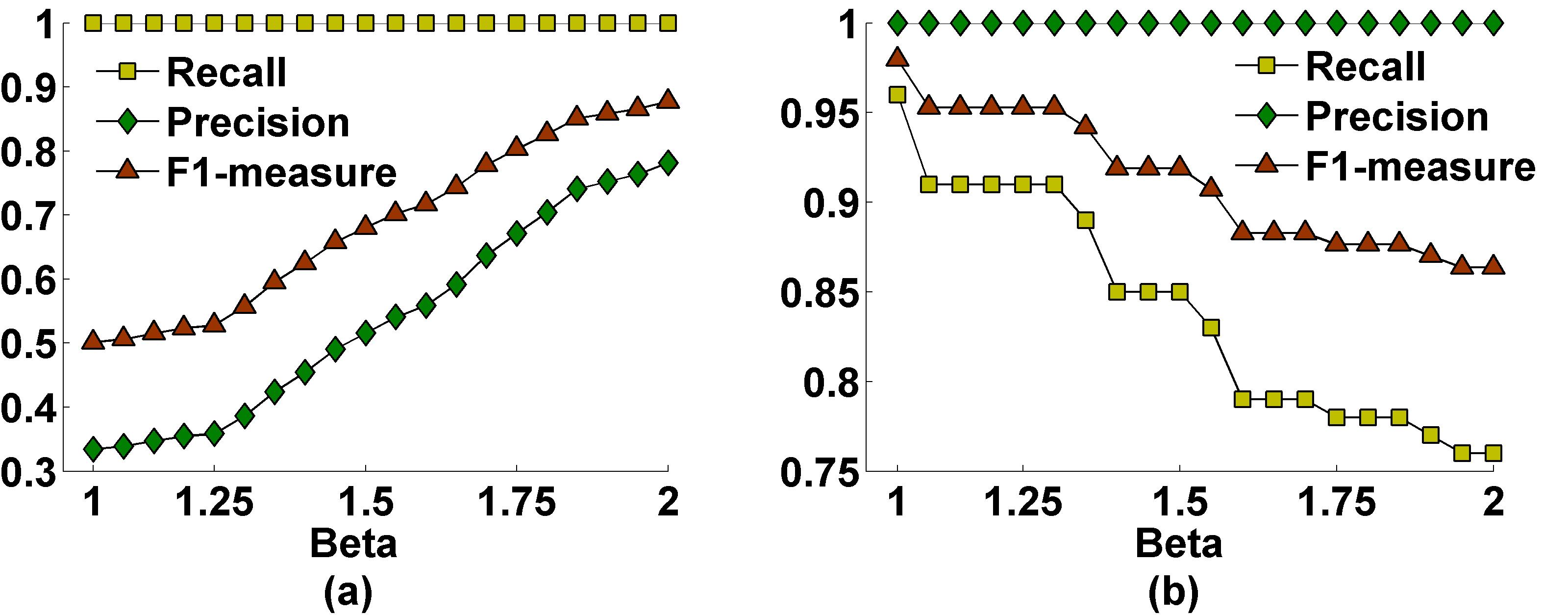}
\caption{The variations of $Recall$,$Precision$ and $F1$-$measure$ along with $\beta$ values. (a) and (b) demonstrate the  variations in $FT$ approach and $FT$-$X$ approach respectively.}
\label{BetaSetting}
\end{figure} 

Although the extended version of \textit{Shewhart control charts} is capable of identifying failures adaptively, it's still necessary to determine two parameters, namely the sliding window $N^{'}$ and $\epsilon$ in order to obtain an optimal detection result. Figure 7$\sim$10 demonstrate the $Recall$, $Precision$,$F1$-$measure$ and $ATTF$ variations along with $\epsilon$ and $N^{'}$ respectively. The variation zone is organized as 10x14 mesh grid. From Figure 7, we observe that in the area where $2 \le N^{'} \le 6$ and $4 \le \epsilon \le 7$, some values are 0 (i.e. $N_{tp}=0$) as there are no deviations exceeding the threshold $\epsilon$. Accordingly,  the $Precision$ and $F1$-$measure$ are 0 too. But in other areas, all the failure points are detected (i.e. $Recall=1$). Thus $F1$-$measure$ changes consistently with $Precision$.  Here we choose the optimal result when $N=6$ and $\epsilon=6.5$ according to $F1$-$measure$. At this point, $Recall=1$,$Precision=0.99$, \textit{F1-measure}=0.995 and $ATTF=6$. 
\begin{figure}[!t]
\begin{minipage}[t]{0.5\linewidth}
\centering
\includegraphics[width=1.6in,height=1.2in]{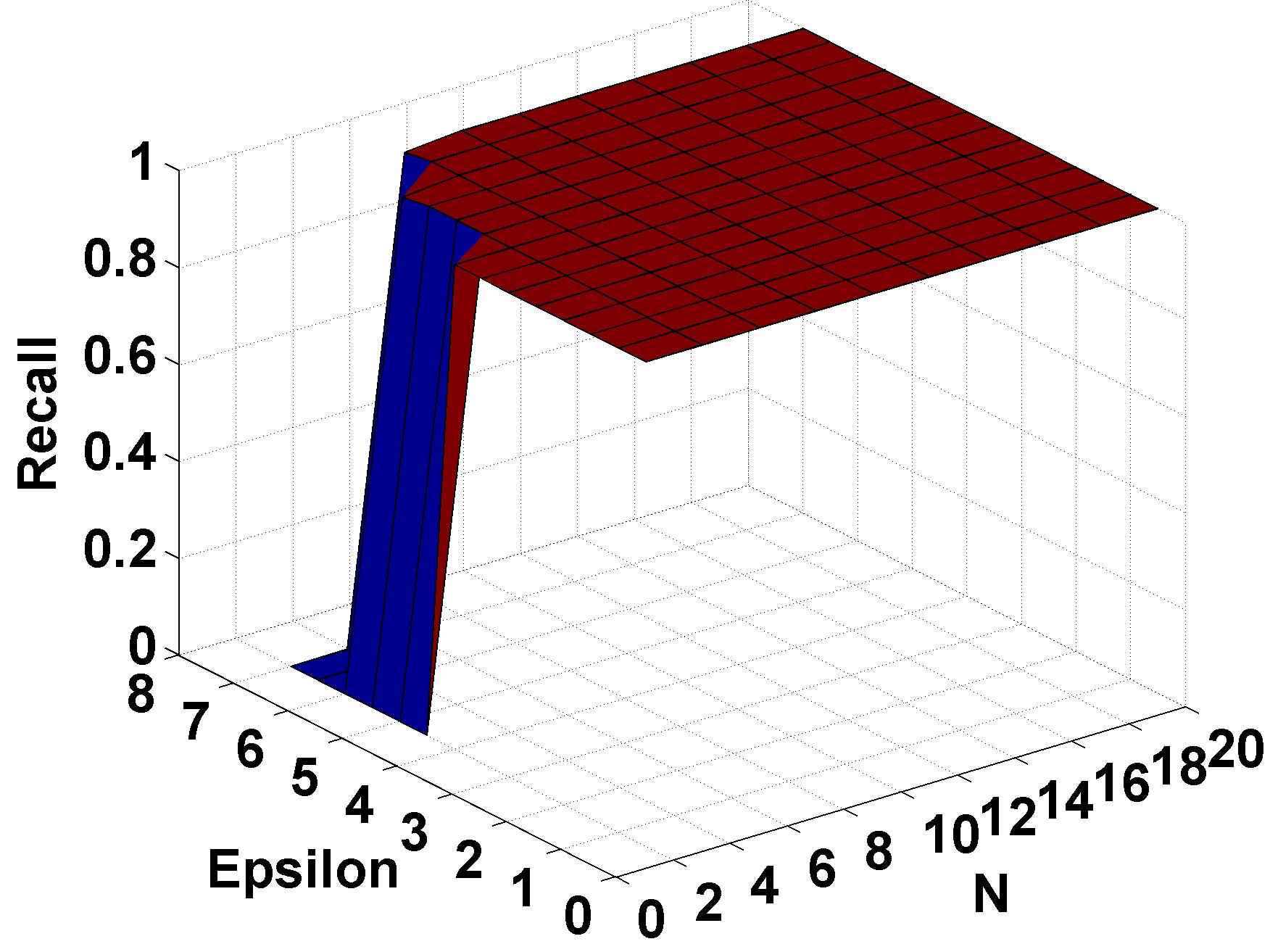}
\caption{$Recall$ variations}
\label{Recall}
\end{minipage}%
\begin{minipage}[t]{0.5\linewidth}
\centering
\includegraphics[width=1.6in,height=1.2in]{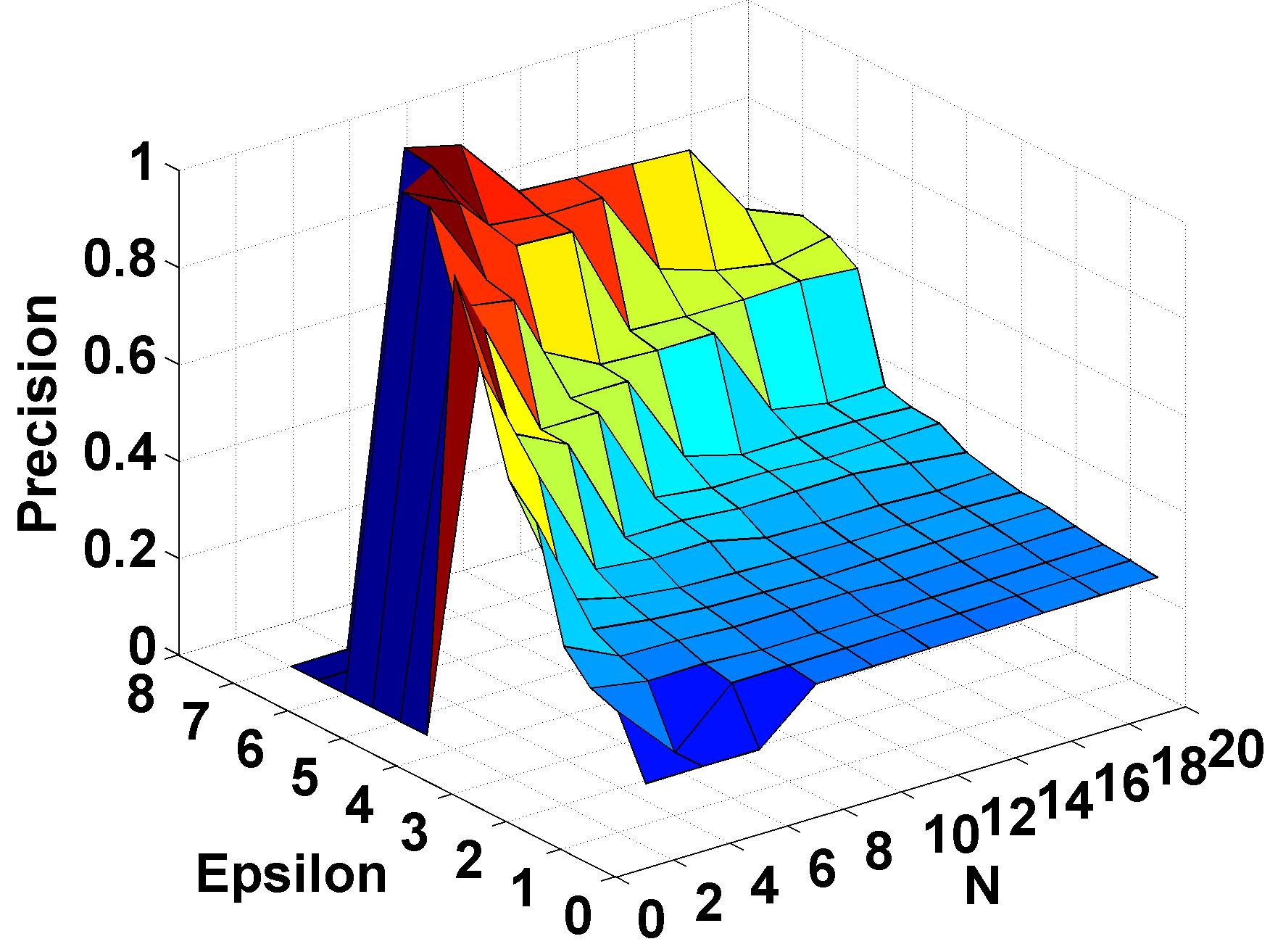}
\caption{$Precision$ variations}
\label{Precision}
\end{minipage}
\end{figure}
\begin{figure}[!t]
\begin{minipage}[t]{0.5\linewidth}
\centering
\includegraphics[width=1.6in,height=1.2in]{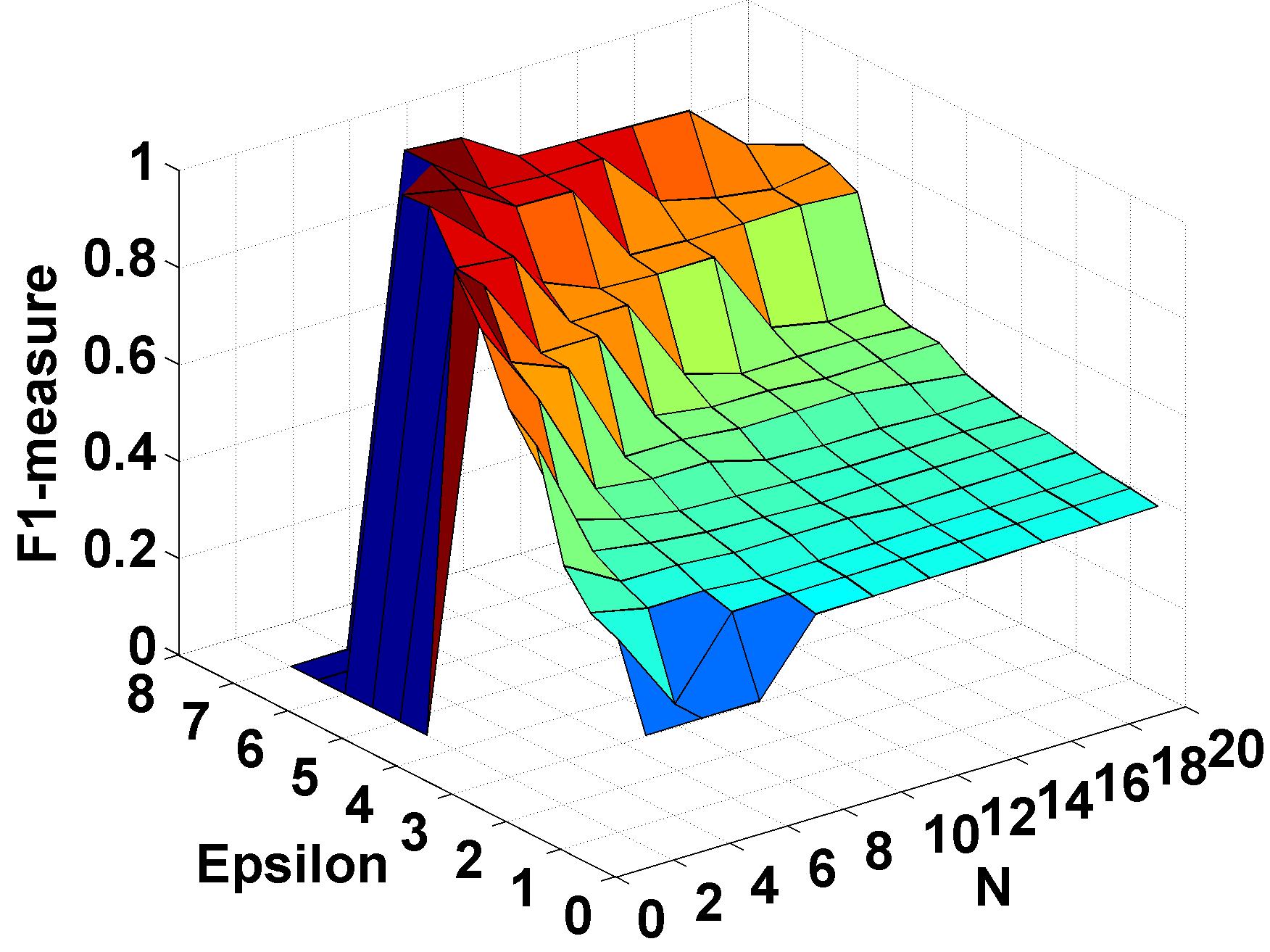}
\caption{$F1$-$measure$ variations}
\label{$F1$-$measure$}
\end{minipage}%
\begin{minipage}[t]{0.5\linewidth}
\centering
\includegraphics[width=1.6in,height=1.2in]{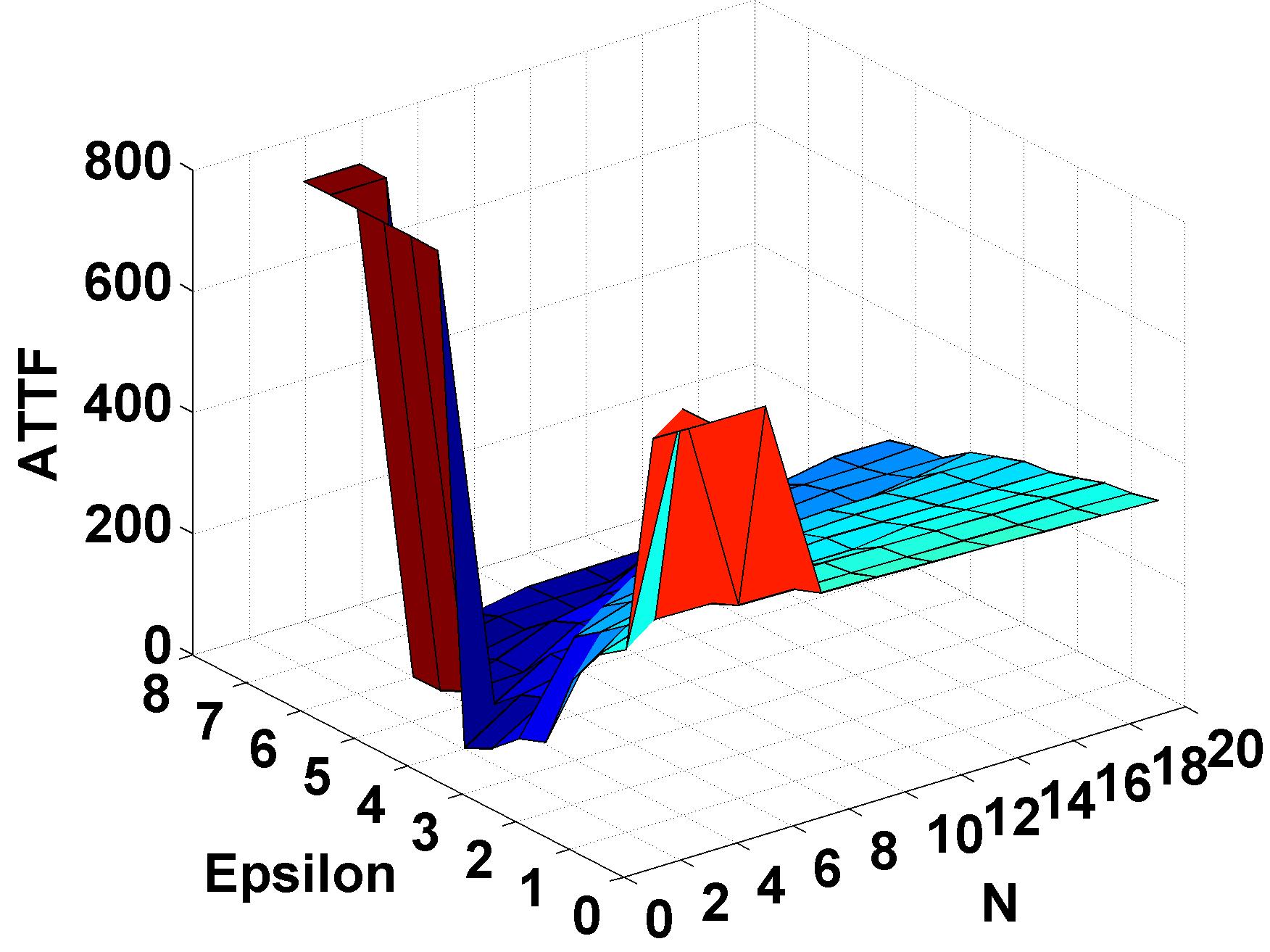}
\caption{$ATTF$ variations}
\label{ATTF}
\end{minipage}
\end{figure}

In the following experiments, we will compare the detection results of $FT$, $FT$-$X$ and the extended version of \textit{Shewhart control charts}  when they achieve the optimal results in the Helix Server system and the real-world AntVision system. In different systems, we will determine the optimal results for different approaches separately. 

Figure 11 depicts the comparisons  of the  failure detection results obtained by $FT$, $FT$-$X$ and the extended \textit{Shewhart control charts} in Helix Sever system. %In this paper, $F1$-$measure$ is leveraged as a critical criterion to compare these detection approaches.  
From Figure 11.(a), we observe that the extended version of  \textit{Shewhart control chart} achieves the best result, $F1$-$measure$=0.995; $FT$-$X$ achieves the second best result, $F1$-$measure$=0.9795; $FT$ achieves the worst result, $F1$-$measure$=0.8899. The detection results of  the extended \textit{Shewhart control chart} and  $FT$-$X$ have about 0.1 improvement compared to the one of $FT$.  Meanwhile, a lower $ATTF$ is obtained by the adaptive approaches such as  $FT$-$X$, shown in Figure 11.(b). A lower  $ATTF$  not only guarantees the failure could be detected in time but also reduces the excessive maintenance cost. Via these comprehensive comparisons, we find that based upon MMSE, the adaptive approaches outperform the statical approaches due to their adaptation to the ever changing runtime environment. 
\begin{figure}[!t]
\begin{minipage}[t]{0.5\linewidth}
\centering
\includegraphics[width=1.8in,height=1.4in]{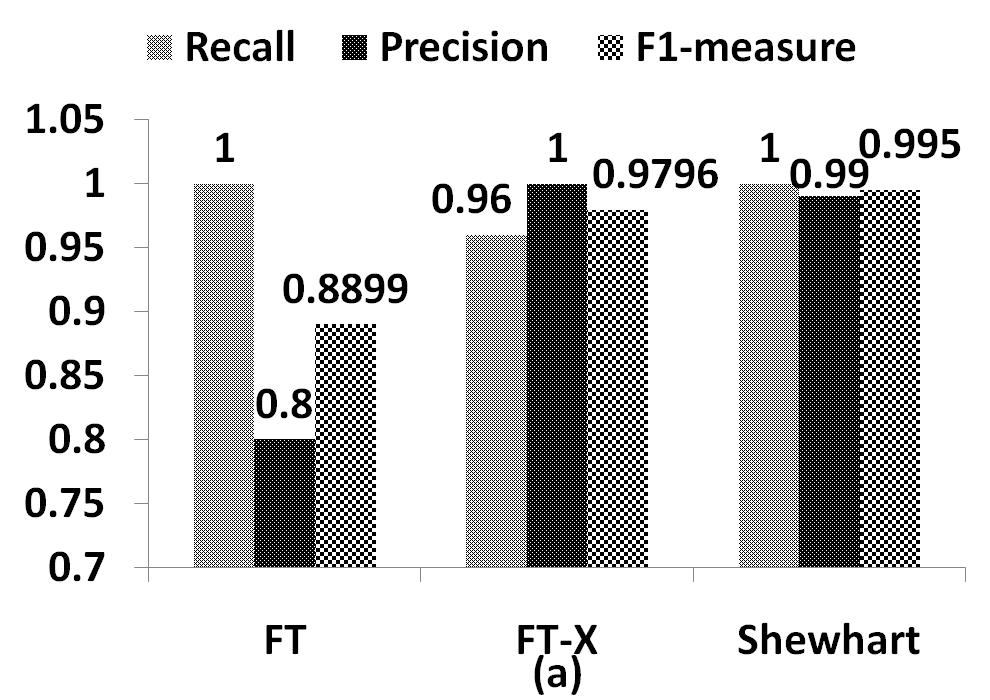}
%\caption{(a)}
\label{Recall}
\end{minipage}%
\begin{minipage}[t]{0.5\linewidth}
\centering
\includegraphics[width=1.5in,height=1.4in]{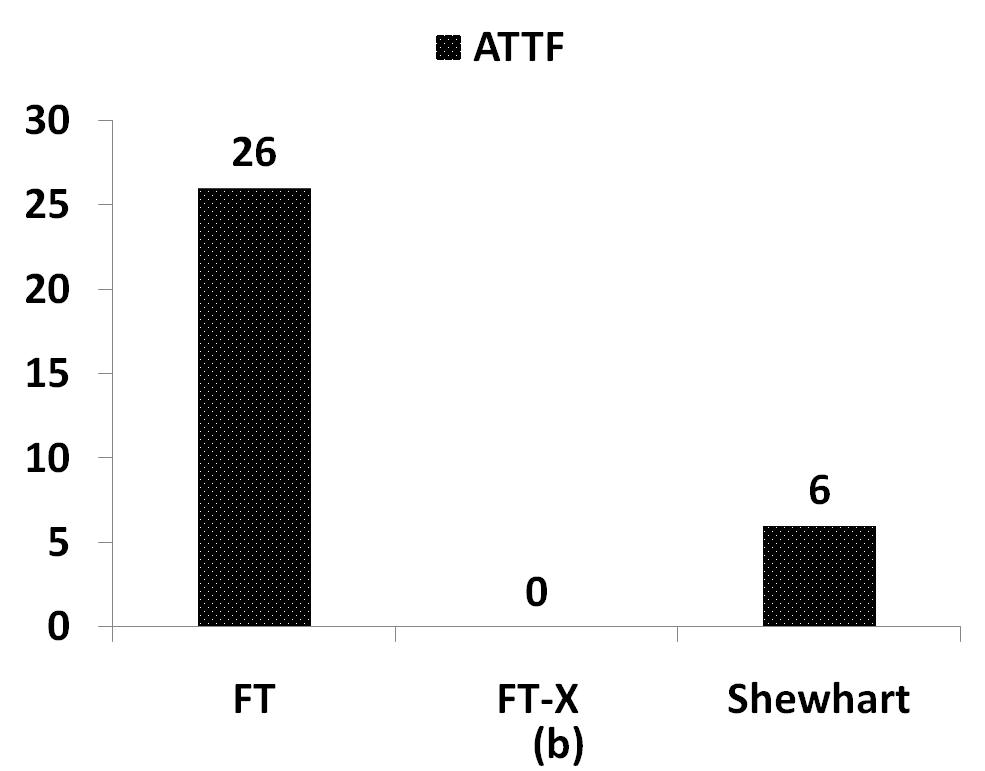}
%\caption{(b)}
\label{Precision}
\end{minipage}
\caption{The comparisons  of the  failure detection results obtained by $FT$, $FT$-$X$ and \textit{Shewhart control charts} in Helix Sever system.  (a) presents $Recall$, $Precision$ and $F1$-$measure$ comparisons and (b) presents $ATTF$ comparisons.}
\end{figure}

Figure 12 shows one slice of MMSE time series in the range $1100\sim1320$ calculated by MMSE algorithm on the performance metrics collected in AntVision system and the optimal failure reports generated by $FT$, $FT$-$X$ and  \textit{Shewhart control chart}. The failure reports generated by $FT$, $FT$-$X$ and \textit{Shewhart control chart} fall  in the range $1213\sim1320$, $1217\sim1320$ and $1219\sim1320$ respectively. It is intuitively observed that \textit{Shewhart control chart} approach achieves the best detection result as almost all its failure reports fall in the decision window. However the detection results achieved by $FT$ and $FT$-$X$ are very similar. This is because there are no significant changes for  MMSE in the range $1000 \sim 1220$, which results in the optimal threshold determined by $FT$ and $FT$-$X$ are very similar, namely  0.233 and 0.4 respectively.  Figure 13 demonstrates the comparisons of failure detection results in terms of $Recall$,$Precision$, $F1$-$measure$ and $ATTF$.  The results also tell us that the adaptive approach based upon MMSE indicator is capable of achieving a better detection accuracy and a lower $ATTF$. To make a broad comparison with the approaches based upon other aging indicators, we conduct the following experiments.  
\begin{figure}[!t]
\centering
\includegraphics[width=3.4in,height=1.4in]{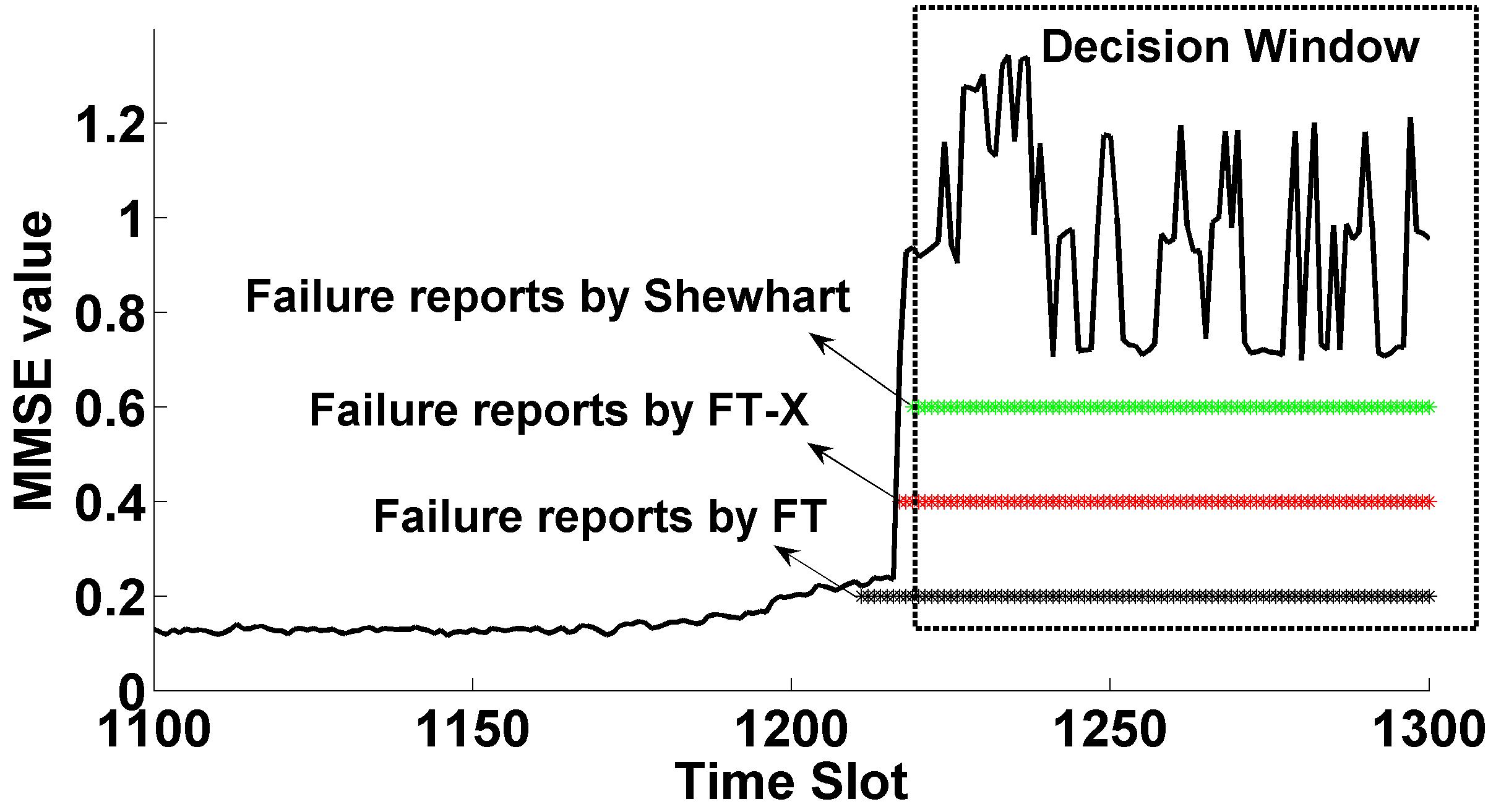} 
\caption{One slice of MMSE data and the failure reports generated by $FT$, $FT$-$X$ and \textit{Shewhart control chart} in AntVision system.}
\label{AntDataDetectionResult}
\end{figure} 
\begin{figure}[!t]
\begin{minipage}[t]{0.5\linewidth}
\centering
\includegraphics[width=1.8in,height=1.4in]{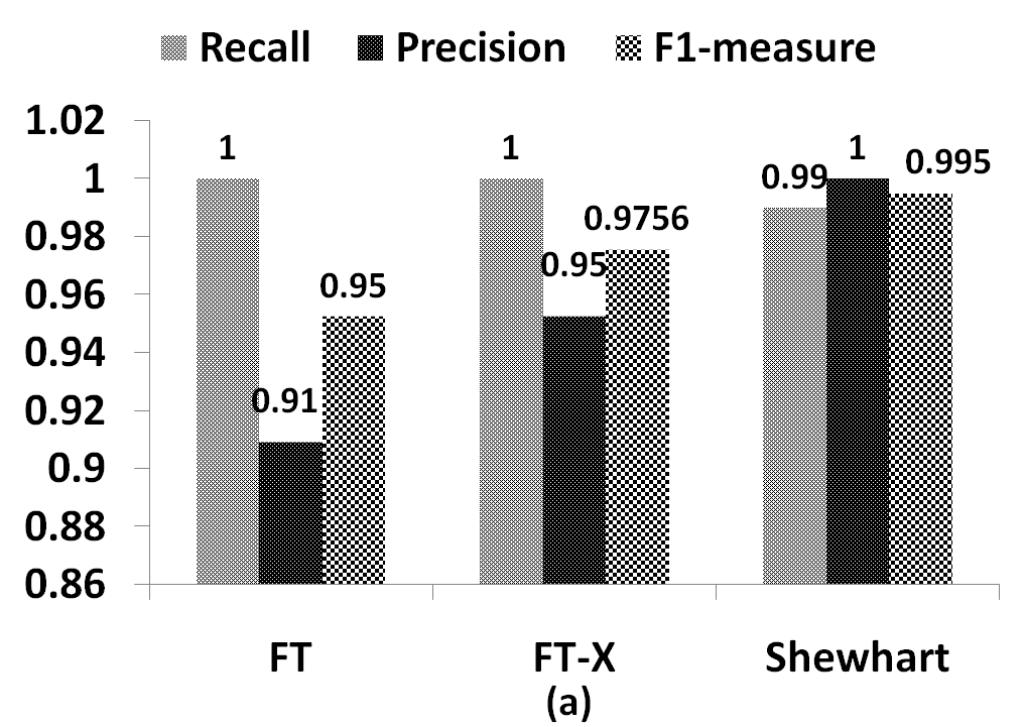}
%\caption{(a)}
\label{AntComparison}
\end{minipage}%
\begin{minipage}[t]{0.5\linewidth}
\centering
\includegraphics[width=1.5in,height=1.4in]{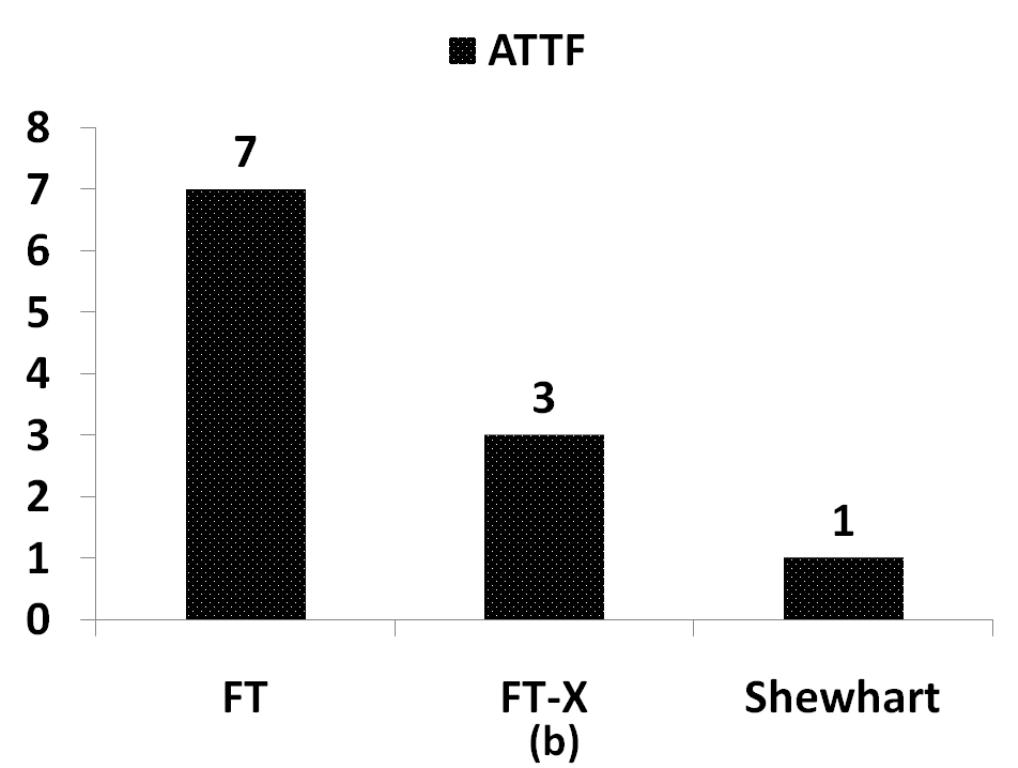}
%\caption{(b)}
\label{AntATTFComparison}
\end{minipage}
\caption{The comparisons  of the  failure detection results obtained by $FT$, $FT$-$X$ and \textit{Shewhart control charts} in AntVision system.}  %(a) presents $Recall$, $Precision$ and $F1$-$measure$ comparisons and (b) presents $ATTF$ comparisons.}
\end{figure}
%
%\begin{figure}[!t]
%\centering
%\includegraphics[width=3.3in]{MMSE-comparison}
%\caption{The variations of $Recall$,$Precision$ and $F1$-$measure$ along with $\beta$ values. (a) and (b) demonstrate the  variations in $FT$ approach and $FT$-$X$ approach respectively.}
%\label{MMSE-comparison}
%\end{figure} 
\subsection{Comparison}
In this section, we will compare the failure detection results obtained by the approaches based upon MMSE and the approaches based upon other explicit or implicit indicators. In previous studies, QoS metrics (e.g. response time, throughout) or runtime performance metrics (e.g. CPU utilization) are more often than not adopted as explicit aging indicators. Accordingly, we adopt \textit{AverageBandwidth} as an explicit aging indicator in Helix Server system and  \textit{CPU utilization} as an explicit aging indicator in AntVision system. $H \ddot{o}lder$ exponent mentioned in \cite{24} is adopted as an implicit aging indicator in these two systems. For different aging indicators, the failure detection approaches vary a little. For \textit{AverageBandwidth}  and  $H \ddot{o}lder$ exponent indicators, we employ a lower boundary test in the threshold based approach and  the extended version of \textit{Shewhart control chart} proposed in \cite{24} in the time series approach both of which are depicted in Section \uppercase\expandafter{\romannumeral 5}.D,  due to their downtrend characteristics. It's worth noting that $\beta$  should vary in the same range e.g. 1$\sim$20 in this paper for $FT$ and $FT$-$X$ in order to conduct fair comparisons.  All of comparisons are conducted in the situations when these failure detection approaches achieve optimal results. 

We first determine the optimal conditions when these approaches achieve their optimal results in Helix Server system. Table \uppercase\expandafter{\romannumeral 1} demonstrates these optimal conditions. Figure 14 shows the comparison results for different indicators in terms of  $Recall$, $Precision$, $F1$-$measure$ and  $ATTF$ respectively. 
\begin{table}[!hbp]
\centering
\caption{The optimal conditions for different approaches based upon different aging indicators in Helix Server system}
\begin{tabular}{c|c|c|c}
\toprule
   & FT &	FT-X  & Shewhart control chart \\
 \midrule
 AverageBandwidth & $\beta=1.8$ &	$\beta=1.8$  &	$N^{'}$=440,$\epsilon=8$ \\
 \hline
 MMSE	 & $\beta=2$ & $\beta=1.1$ & $N^{'}$=4,$\epsilon=6$\\
\hline
 $H \ddot{o}lder$ 	 & $\beta=5.3$ & $\beta=5.3$ & $N^{'}$=40,$\epsilon=5$ \\
\bottomrule  
\end{tabular}
\end{table}

From Figure 14.(a), we observe that the extended version of \textit{Shewhart control chart} approach achieves an ideal recall (i.e. $Recall=1$) no matter which indicator is chosen. However for $FT$ and $FT$-$X$ approaches, the detection result heavily depends on aging indicators.  The $Recall$ of $FT$ and $FT$-$X$ based upon MMSE are 1 and 0.91 respectively, much higher than the results obtained by the approaches based upon \textit{AverageBandwidth}, 0.52 and  $H \ddot{o}lder$, 0.62. The effectiveness  of MMSE  is even more significant than  the other two indicators in term of $Precision$. We observe that the $Precision$ of  failure detection approaches based upon MMSE  is  up to 9 times higher than the one of $FT$ or $FT$-$X$ based upon $H \ddot{o}lder$, and 5 times higher than the one of $FT$ or $FT$-$X$ based upon \textit{AverageBandwidth}, shown in Figure 14.(b). Accordingly, the MMSE is much more powerful to detect AOF than $H \ddot{o}lder$ and \textit{AverageBandwidth} in $F1$-$measure$ demonstrated in Figure 14.(c).  From the point of view of $ATTF$, the approaches based upon MMSE obtain up to 3 orders of magnitude improvement than the ones based upon the other two indicators. For example in Figure 14.(d), for $FT$-$X$  approach,  the $ATTF$  based upon \textit{AverageBandwidth} and  $H \ddot{o}lder$ are 1570 and 1700 respectively, but the  $ATTF$ based upon MMSE is 0. The extraordinary effectiveness of MMSE is attributed to its three properties: \textit{monotonicity}, \textit{stability} and \textit{integration}. However, the single runtime parameter e.g. \textit{AverageBandwidth} can't comprehensively reveal the aging state of the whole system and the fluctuations involved in this indicator result in  much detection bias. Figure 15 shows a representative $AverageBandwidth$  variations from system start to  ``failure''.  We observe that the  \textit{AverageBandwidth} may be low even at normal state. The $H \ddot{o}lder$  exponent indicator also suffers from this problem. Although a downtrend indeed exists in $H \ddot{o}lder$ exponent indicator indicating the complexity is increasing which is compliant with the result in \cite{24}, shown in Figure 16, the  instability hinders to achieve a high accurate failure detection result. From above comparisons, we find out the detection results obtained by $FT$ and $FT$-$X$ based upon $AverageBandwidth$ or  $H \ddot{o}lder$ are the same. That's because the minimum point of the aging indicator is involved simultaneously in the training data of $FT$ and $FT$-$X$ demonstrated in Figure 15 and Figure 16. Therefore the optimal threshold values calculated by $FT$ and $FT$-$X$ are the same.  
\begin{figure}[!t] 
\centering  
\includegraphics[width=3.4in]{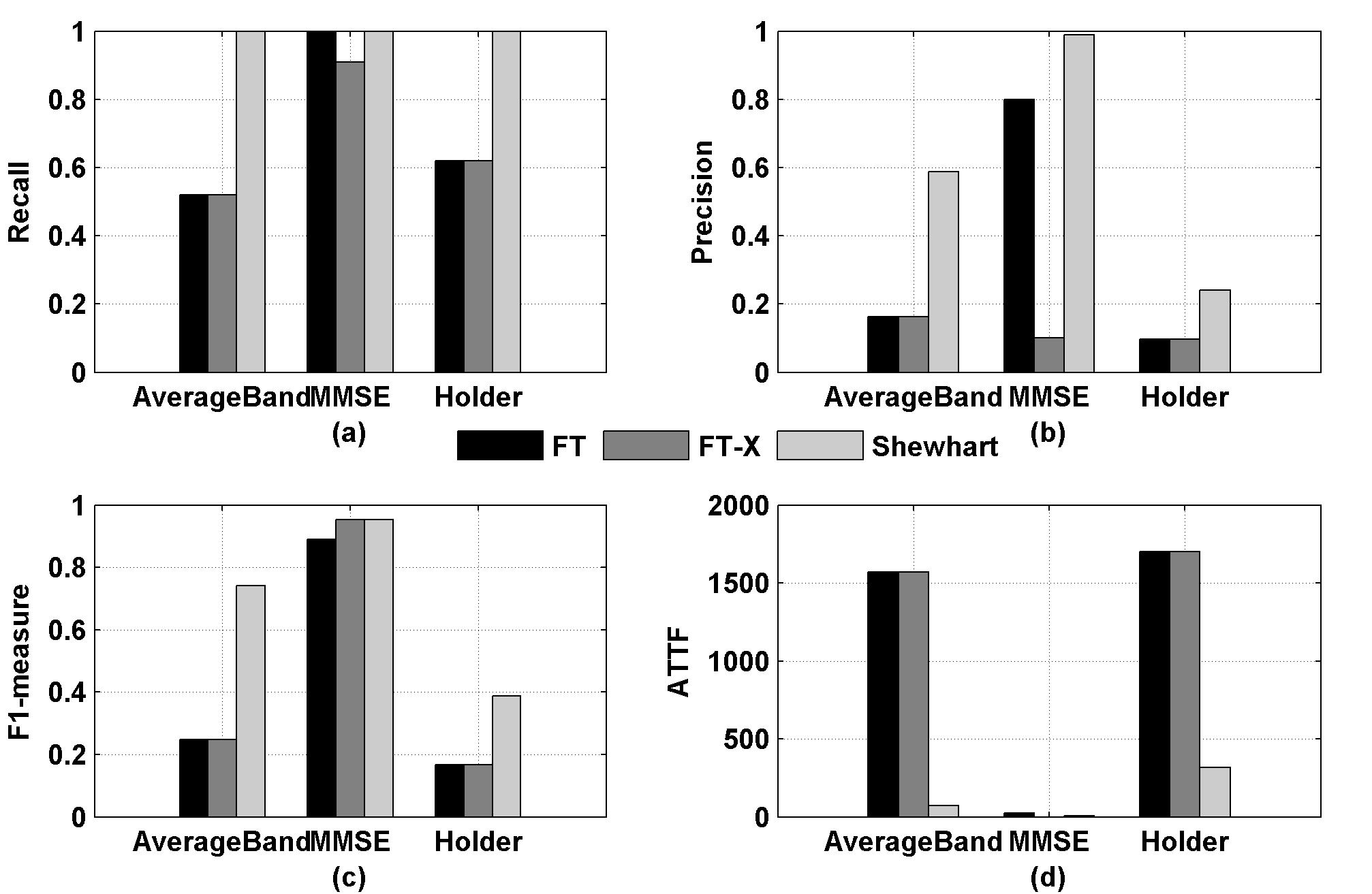} 
\caption{The comparison results of the detection approaches based upon different aging indicators in Helix Server system. Here ``AB" is short for \textit{AverageBandwidth}.}
\label{HelixServerComparison}
\end{figure} 
%\begin{figure}[!t] 
%\centering  
%\includegraphics[width=3.3in,height=1.6in]{Comparison_recall} 
%\caption{The comparison result for different indicators in $Recall$.}
%\label{Comparison_recall}
%\end{figure} 
%\begin{figure}[!t]
%\centering
%\includegraphics[width=3.3in,height=1.6in]{Comparison_precision} 
%\caption{The comparison result for different indicators in $Precision$.}
%\label{AntDataDetectionResult}
%\end{figure} 
%\begin{figure}[!t]
%\centering
%\includegraphics[width=3.3in,height=1.6in]{Comparison_F1_measure} 
%\caption{The comparison result for different indicators in $F1$-$measure$.}
%\label{Comparison_F1_measure}
%\end{figure} 
%\begin{figure}[!t]
%\centering
%\includegraphics[width=3.3in,height=1.6in]{Comparison_ATTF} 
%\caption{The comparison result for different indicators in $ATTF$.}
%\label{Comparison_ATTF}
%\end{figure} 
\begin{figure}[!t]
\centering
\includegraphics[width=3.4in,height=1.6in]{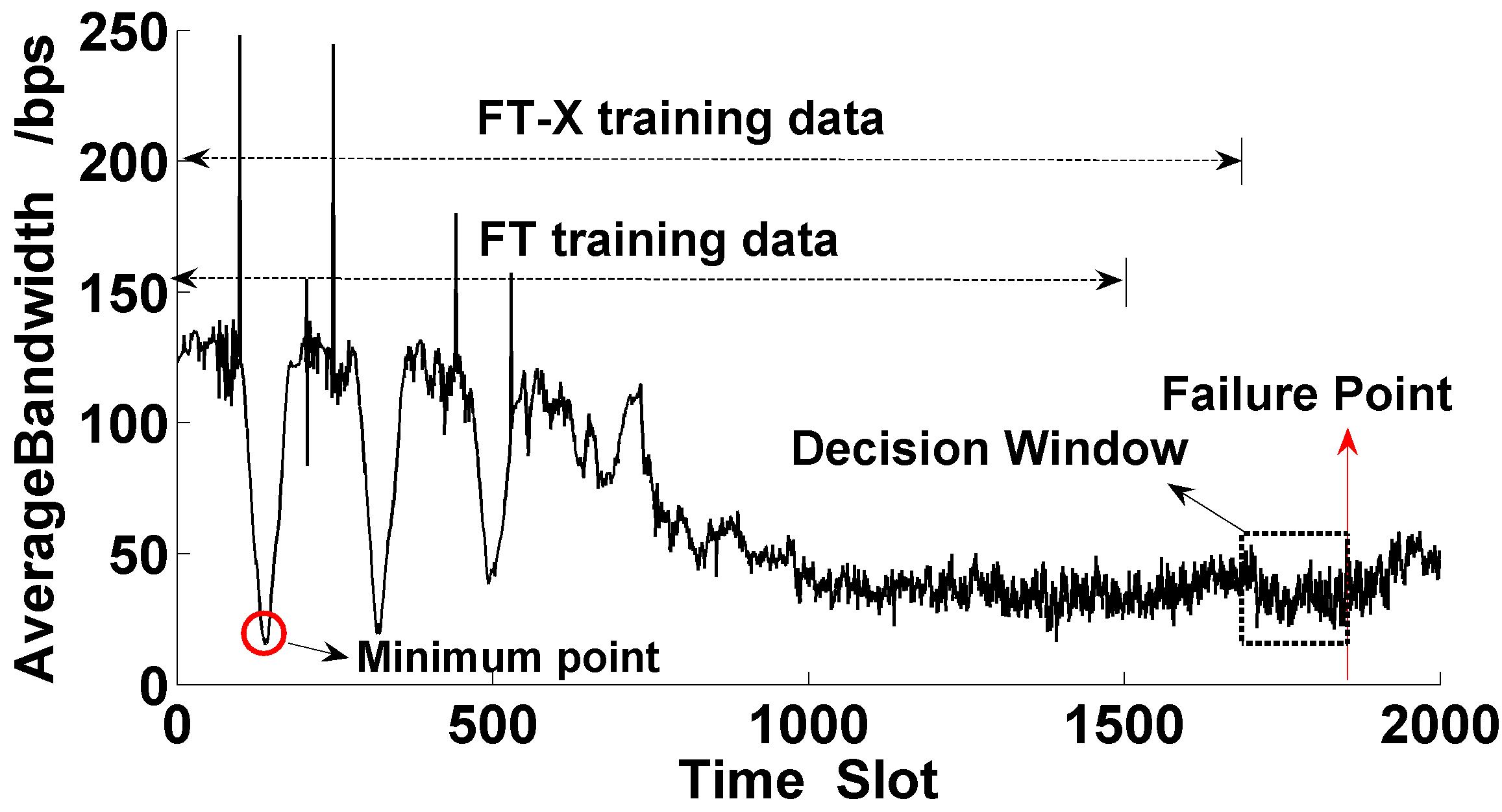} 
\caption{The \textit{AverageBandwidth} data from system start to ``failure".}
\label{ReasonOfNoDifferenceFT}
\end{figure} 
\begin{figure}[!t]
\centering
\includegraphics[width=3.4in,height=1.6in]{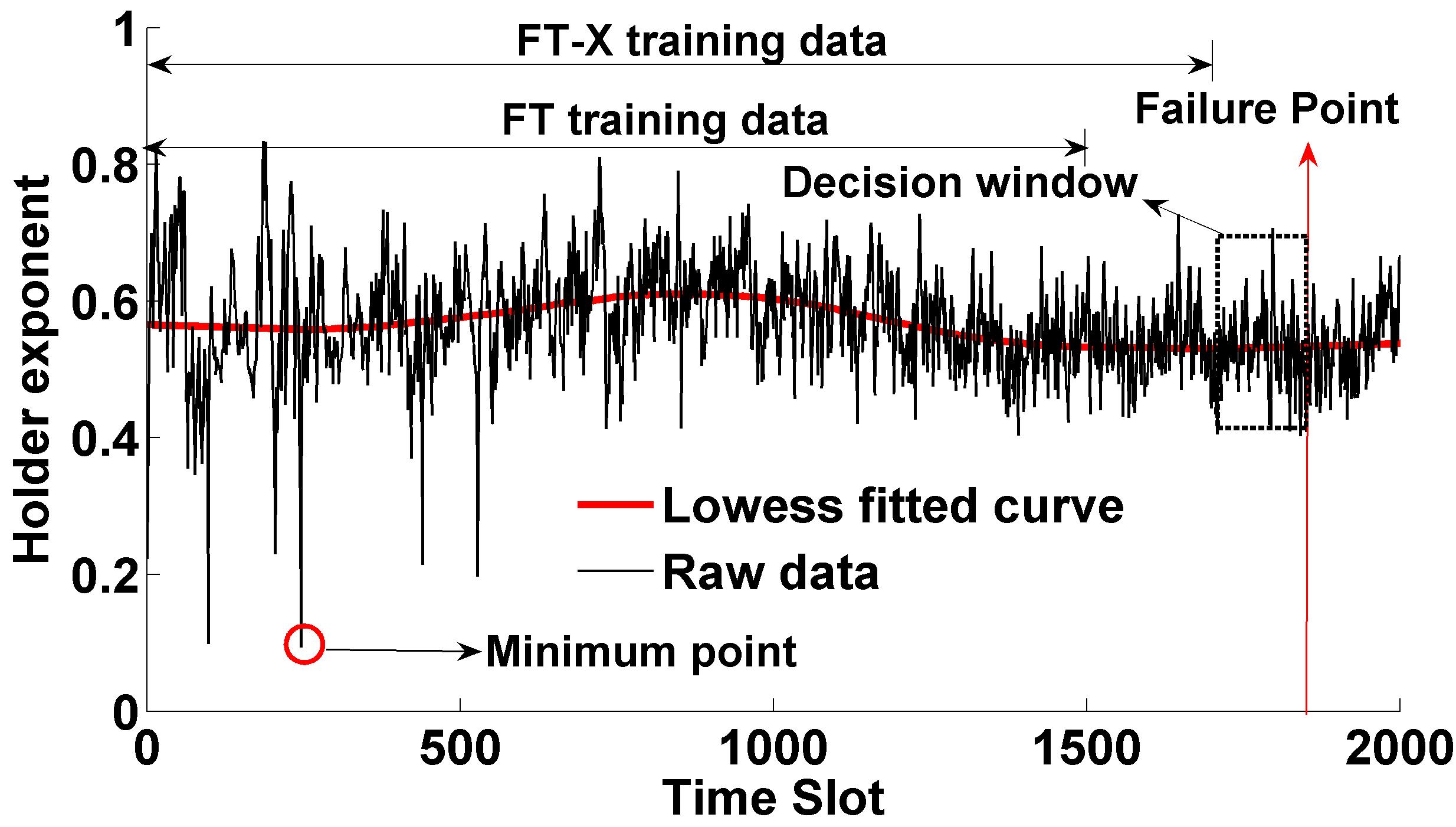} 
\caption{The  $H \ddot{o}lder$  data from system start to  ``failure". The curve fitted by $Lowess$ \cite{40} is used to present the downtrend. }
\label{HelixHolder}
\end{figure} 

The optimal conditions for these failure detection approaches based upon \textit{CPU Utilization}, \textit{MMSE} and $H \ddot{o}lder$ exponent in AntVision system are listed in Table \uppercase\expandafter{\romannumeral 2}.  
\begin{table}[!hbp]
\centering
\caption{The optimal conditions for different approaches based upon different aging indicators in AntVision system}
\begin{tabular}{c|c|c|c}
\toprule
   & FT &	FT-X  & Shewhart control chart \\
 \midrule
 CPU Utilization & $\beta=1$ &	$\beta=-$  &	$N^{'}$=75,$\epsilon=17$ \\
 \hline
 MMSE	 & $\beta=2$ & $\beta=1.3$ & $N^{'}$=8,$\epsilon=7$\\
\hline
 $H \ddot{o}lder$ 	 & $\beta=4.0$ & $\beta=19$ & $N^{'}$=165,$\epsilon=8$ \\
\bottomrule  
\end{tabular}
\end{table}
An interesting finding is that the optimal condition of $FT$-$X$ based upon \textit{CPU Utilization} indicator is $\beta=-$ which means  we can't find an optimal $\beta$ in the range $1\sim20$.  By investigating the detection results, we observe that the $Recall$, $Precision$ and $F1$-$measure$ are all 0 no matter  which value $\beta$ is chosen in the range $1\sim20$. Figure 17 provides the reason why we get this observation. The maximum CPU utilization involved in the training data in $FT$-$X$  falling in  the range $1\sim1200$, exceeds all the CPU Utilization in  the decision window.  Therefore according to the threshold calculated by $FT$-$X$, we can't detect any failures (i.e. $N_{tp}=0$).  While for $FT$ approach, the maximum CPU utilization in the training data is lower than the maximum CPU Utilization in the decision window.  Hence some failure points can be detected by $FT$. This is the reason why $FT$ outperforms $FT$-$X$ based upon \textit{CPU Utilization} in AntVision system. And this could be regarded as a drawback of  non-monotonicity of the \textit{CPU Utilization} indicator. 
\begin{figure}[!t]
\centering  
\includegraphics[width=3.4in]{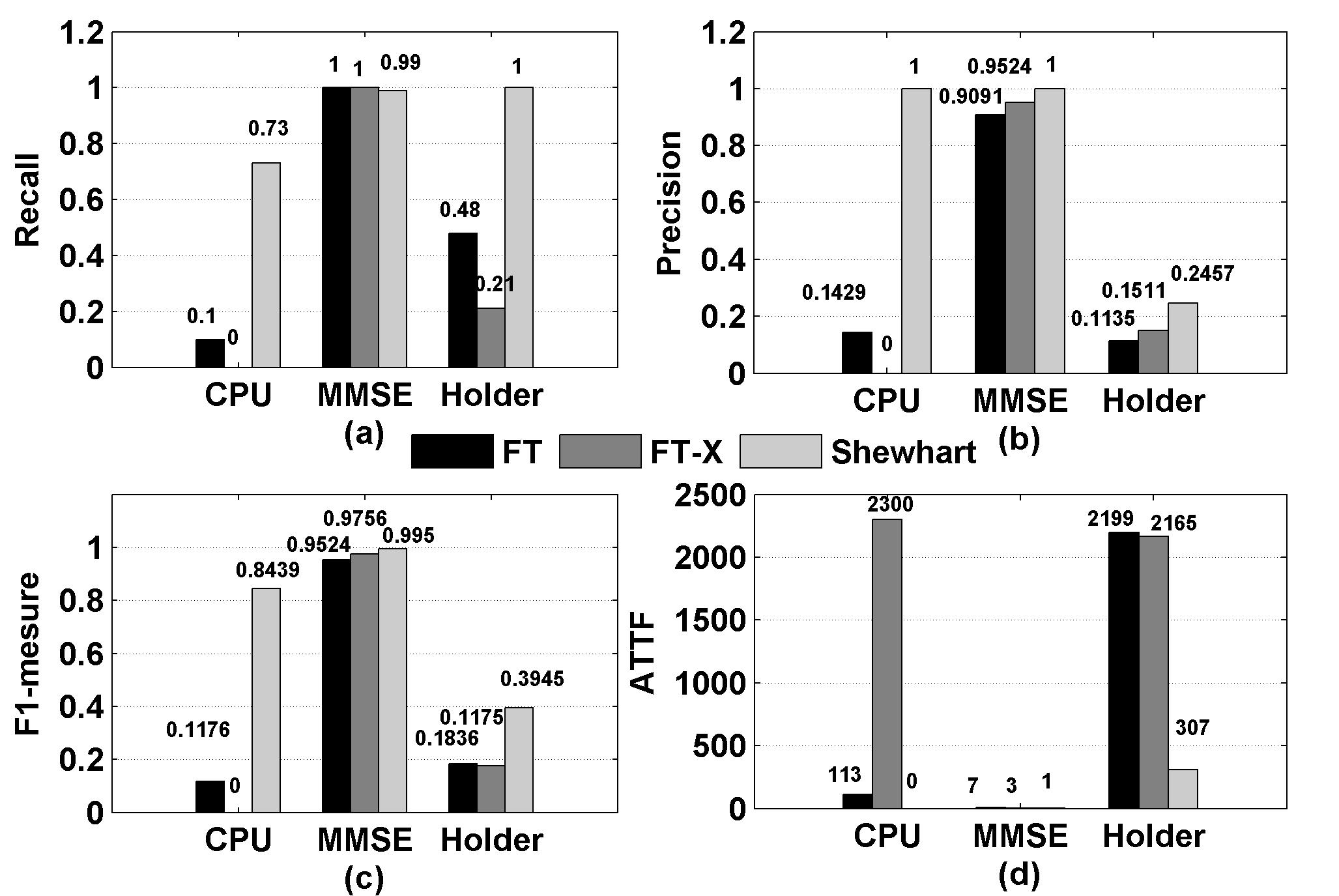} 
\caption{The CPU utilization and corresponding MMSE data in AntVision system.}
\label{AntDataComparisonNew}
\end{figure} 

Figure 18 demonstrates the comparison results in terms of $Recall$,$Precision$,$F1$-$measure$ and $ATTF$ amongst the failure detection approaches based upon different aging indicators in AntVision system. From this figure, we observe that the $F1$-$measure$  achieved by  MMSE-based approaches are higher than 0.95 and much better than the one achieved by \textit{CPU Utilization}-based and $H \ddot{o}lder$ exponent-based approaches. Meanwhile, the $ATTF$ is  significantly reduced from a large number (e.g. 2300) to a very tiny number (e.g. 1) by MMSE-based approaches. We also observe that the extended version of \textit{Shewhart control chart} approach performs better than the other two approaches no matter which indicator is chosen. 

Finally, through comprehensive comparisons above, we conclude that  MMSE-based approaches extraordinarily outperform an explicit indicator (i.e. \textit{CPU Utilization}) based approach and an implicit indicator (i.e. $H \ddot{o}lder$ exponent) based approach. The high  accuracy of MMSE results from its three properties: \textit{Monotonicity},\textit{Stability},\textit{Integration}. And based upon MMSE, the adaptive  detection approaches i.e. the extended version of \textit{Shewhart control chart} performs better. 
\begin{figure}[!t]
\centering  
\includegraphics[width=3.4in]{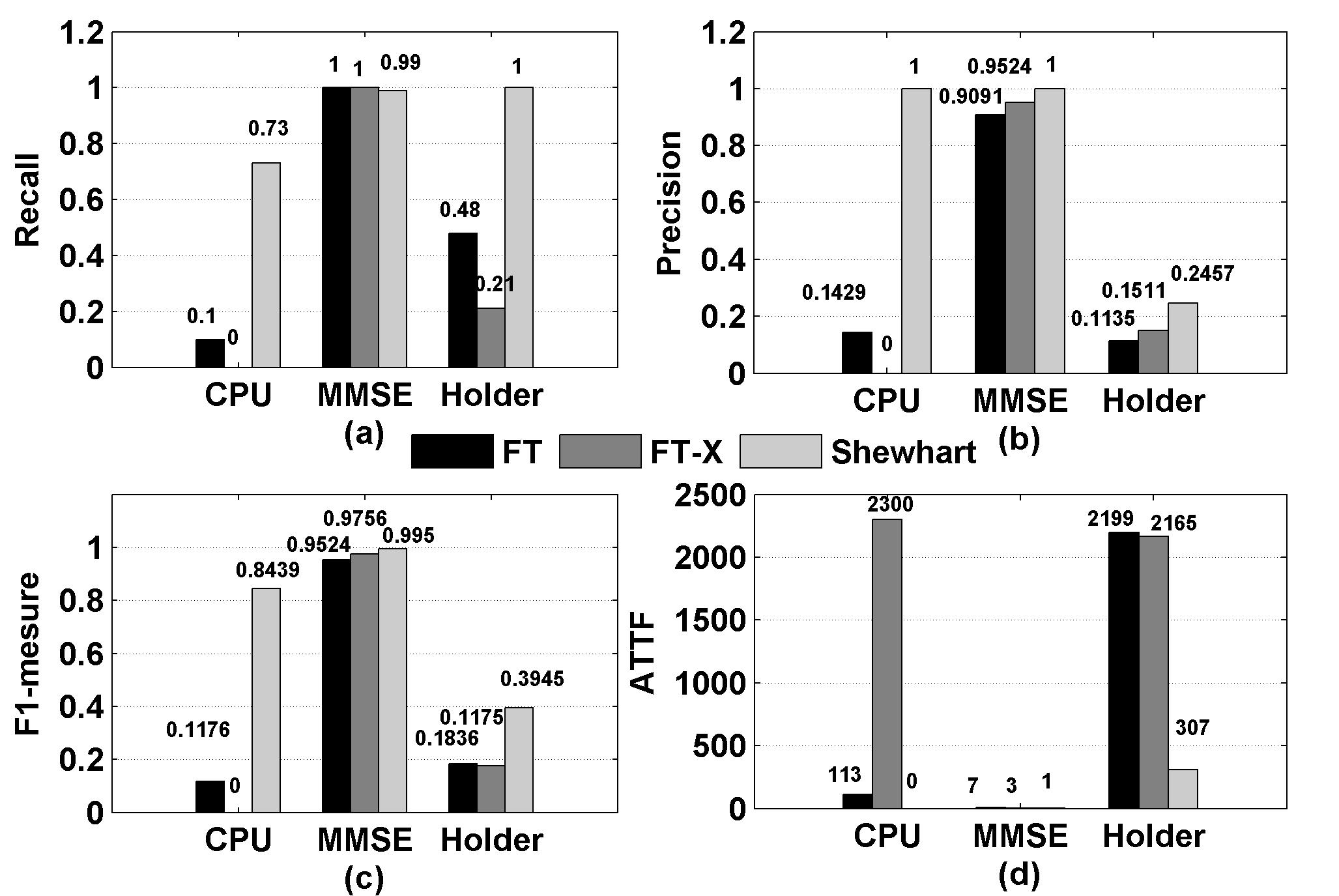} 
\caption{The comparison results of the detection approaches based upon different aging indicators in AntVision system. Here ``CPU" means CPU utilization.}
\label{AntDataComparisonNew}
\end{figure} 
%\begin{figure}[!t]
%\centering  
%\includegraphics[width=3.3in,height=1.6in]{AntRecallComparison} 
%\caption{The comparison result for different indicators in $Recall$.}
%\label{AntRecallComparison}
%\end{figure} 
%\begin{figure}[!t]
%\centering
%\includegraphics[width=3.3in,height=1.6in]{AntPrecisionComparison} 
%\caption{The comparison result for different indicators in $Precision$.}
%\label{AntPrecisionComparison}
%\end{figure} 
%\begin{figure}[!t]
%\centering
%\includegraphics[width=3.3in,height=1.6in]{AntF1MeasureComparison} 
%\caption{The comparison result for different indicators in $F1$-$measure$.}
%\label{AntF1MeasureComparison}
%\end{figure} 
%\begin{figure}[!t]
%\centering
%\includegraphics[width=3.3in,height=1.6in]{AntATTFComparison2} 
%\caption{The comparison result for different indicators in $ATTF$. The bars of $ATTF$ based upon MMSE are not shown in this figure due to their tininess. $ATTF=2300$ means no failures are detected.}
%\label{AntATTFComparison2}
%\end{figure} 
\subsection{Overhead}
The whole analysis procedure of CHAOS except data collection is conducted on a separate machine. Hence it causes very little resource footprint on a test or production system. To evaluate whether CHAOS satisfies the realtime  requirement, we calculate the execution time of the whole procedure. The average execution time of different modules of CHAOS in AntVision system are shown in table \uppercase\expandafter{\romannumeral 3} where MS means Metric selection, MMSE-C means MMSE calculation. Even the most computation-intensive module, namely Metric selection module only consumes 0.875 second and the whole procedure consumes a little more than 1 second. Therefore CHAOS is light-weight enough to satisfy the realtime requirement.
\begin{table}[!hbp]
\centering
\caption{The average execution time of different modules of CHAOS in AntVision system.}
\begin{tabular}{c|c|c|c|c|c}
\toprule
   & MS & MMSE-C  &  FT  & FT-X & Shewhart \\
 \midrule
Time (second) & 0.875 & 0.123	& 0.016  &  0.018  & 0.270  \\
\bottomrule  
\end{tabular}
\end{table}   
\section{Related work}
As the first line of defending software aging, accurate detection of Aging-Oriented Failure is essential. A large quantity of work has been done in this area. Here we briefly discuss related work that has inspired and informed our design, especially work not previously discussed. The related work could be roughly classified into two categories: explicit indicator based method and implicit indicator based method.

\textbf{Explicit indicator based method:} 
The explicit indicator based method usually uses the directly observed performance metrics as the aging indicators and  develops aging detection approaches based upon these indicators. Actually according to our review, most of prior studies such as \cite{0,1,2,3,4,7,8,12,13,14,19,20,50,60} and etc belong to this class. In \cite{0,2,7,19,20}, they treat  system resource usage (e.g. CPU or memory utilization,swap space) as the aging indicator while \cite{3,4,12,13,14,60} take the application specific parameters (e.g. response time, function call) as the aging indicators. Based on these indicators, they detect or predict Aging-Oriented Failure via time-series analysis \cite{0,3,8,12,19,20}, machine learning \cite{4,21,50} or  threshold-based approach \cite{22,23}. The common drawback of these approaches is embodied in the aging indicators' insufficiency due to their weak correlation with software aging.  Hence the detection or prediction results have not reached  a satisfactory level no matter which approaches are adopted. Against this drawback, this paper proposes a new aging indicator,MMSE, which is extracted from the directly observed performance metrics. 

%\textbf{Implicit indicator based method:}
\textbf{Implicit indicator based method:} 
Contrary to the explicit indicator based method, the implicit indicator based method employs aging indicators embedded in the directly observed performance metrics.  These aging indicators are declared to be more sufficient to indicate software aging. Our method falls into this class.  Cassidy, et.al \cite{61}  and Gross, et.al \cite{62} leveraged ``residual" between the actual performance data (e.g. queue length) and the estimated performance data obtained by a multivariate analysis method (e.g. Multivariate State Estimation Technique) as the aging indicator. Then the software's fault detection procedure used a Sequential Probability Ratio Test (SPRT) technique to determine whether the residual value is out of bound. %The similar ``residual" indicator is also adopted in [] and [] with an exception that the estimation model changes from MSET to  Autoregressive Integrated Moving Average (ARIMA). 
Mark, et.al \cite{24} proposed another implicit aging indicator: $H\ddot{o}lder$ exponent. They showed that the $H\ddot{o}lder$ exponent of memory utilization decreased with the degree of software aging. By identifying the second breakdown of $H\ddot{o}lder$ exponent data series through an online Shewhart algorithm, the Aging-Oriented Failure was detected. Although Jia \cite{14} didn't introduce any implicit aging indicator, he showed software aging process was nonlinear and chaotic. Hence, some complexity-related metrics such as entropy, Lyapunov exponent and etc are possible to be aging indicators. And our work is inspired by Mark , et.al \cite{24} and Jia, et.al \cite{14}. However, the prior studies had no quantitative proof about the viability of their implicit aging indicators, no abstraction of the properties that an ideal aging indicator should have and no multi-scale extension. Moreover the effectiveness of $H\ddot{o}lder$ exponent was only evaluated under emulated increasing workload and a thorough evaluation under real workload was absent in the their paper. These defects will result in bias in the detection results, which is shown in the real experiments in section \uppercase\expandafter{\romannumeral 6}.                
%\subsection{Entropy-based Anomaly Detection}

Another implicit indicator is MSE, although it hasn't been employed in software aging analysis before this work. However MSE has been widely used to measure the irregularity variation of pathological data such as electrocardiogram data \cite{26,27,37,44}. Motivated by these studies, we first introduce MSE to software aging area. However, we argue that software aging is a complex procedure affected by many factors. Hence,to accurate measure software aging, a multi-dimensional approach is necessary. We extend the conventional MSE to  MMSE via several modifications.  Wang, et.al \cite{43} also adopts entropy as an indicator of performance anomaly. But he measures the entropy using the traditional Shannon entropy rather than MSE. 
\section{Conclusion}
In this paper, we proposed a novel implicit aging indicator namely MMSE which leverages the complexity embedded in runtime performance metrics to indicate software aging. Through theoretical proof and experimental practice, we demonstrate that entropy increases with the degree of software aging monotonously. To counteract the system fluctuations and comprehensively describe software aging process, MMSE integrates the entropy values extracted from multi-dimensional performance metrics at multiple scales. Therefore, MMSE  satisfies the three properties, namely \textit{Monotonicity}, \textit{Stability},  and \textit{Integration} which we conjecture an ideal aging indicator should have.  
%MMSE leverages the complexity embedded in runtime 
%Through theoretical proof and experimental practice, we have verified that MMSE satisfies  the three properties: \textit{Monotonicity}, \textit{Stability},  and \textit{Integration} which we believe an ideal aging indicator should have. 
Based upon MMSE, we design and develop a proof-of-concept named CHAOS which contains three failure detection approaches, namely $FT$ and $FT$-$X$ and the extended version of \textit{Shewhart control chart}. The experimental evaluation results in a VoD system  and in a real-world production system, AntVision, show that  CHAOS can achieve extraordinarily high accuracy and near 0 $ATTF$. Due to the \textit{Monotonicity} of MMSE, the adaptive approaches such as $FT$-$X$  outperform the static approach such as $FT$ while this is not true for other aging indicators. Compared to previous approaches, the accuracy of failure detection approaches based upon MMSE is increased by  up to 5  times, and the $ATTF$ is reduced by 3 orders of magnitude. In addition,  CHAOS is light-weight enough to satisfy the realtime requirement. We believe  that CHAOS is an indispensable complement to conventional failure detection approaches.

\appendices
\section{Proof of Entropy Increase}
Our proof is based on three basic assumptions:  \\
\textbf{Assumption 1:} The software systems or components  only exhibit binary states during running: working state $s_{w}$ and failure state $s_{f}$.  \\
\textbf{Assumption 2:} The probability of  $s_{f}$   increases monotonously with  the degree of software aging. \\
\textbf{Assumption 3:} If the probability of  $s_{w}$  is less than the probability of $s_{f}$ , the system will be rejuvenated at once. 

A system or a component may exhibit more than two states during running, but here we only consider two states: working  and failure state, which is compliant with the classical three states  i.e. up ,down  and rejuvenation  mentioned in \cite{6,54,55} without considering rejuvenation state.  %If the aging indicator exceeds a preset threshold, it steps into the failure state; otherwise, it stays at the working state. Here a failure state not only refers to system crash but also refers to SLA violation. 
According to the description of software aging stated in the introduction section, the failure rate increases with the degree of software aging. Thus \textbf{Assumption 2} is intuitional. Actually increasing failure probability is also a common assumption in previous studies \cite{45,51,52,53,54,55} in order to obtain an optimal rejuvenation scheduling. For a software system, it's unacceptable if only a half or even less of the total requests are processed successfully especially in modern service oriented systems. A software system  is forced to restart before it enters into a non-service state. Therefore \textbf{Assumption 3} is reasonable. 

If the software system is represented  as a single component, the system entropy at time $t$ is defined as:
\begin{equation}
    E(t)=-(p_{w}(t)*ln(p_{w}(t))+p_{f}(t)*ln(p_{f}(t)))
\end{equation}
where $p_{w}(t)$ and $p_{f}(t)$ represent the probability of normal state $s_{w}$ and failure state $s_{f}$ at time $t$ respectively and $p_{w}(t)+p_{f}(t)=1$. At the initial stage, namely $t=0, p_{w}(0)=1$, we say the system is completely new. At this moment, the entropy $E(t)$ equals 0. As software performance degradation, $p_{w}(t)$ decreases from 1 to 0 while $p_{f}(t)$ increases from 0 to 1. We assume the failure rate $h(t)$  conforms to a Weibull distribution with two parameters which is commonly used in previous studies \cite{48,51,54,55}. The distribution is described as:
\begin{equation}  
h(t)=\frac{\beta }{\alpha }({\frac{t}{\alpha}})^{\beta-1}{e}^{-({\frac{t}{\alpha}})^{\beta}}
\end{equation}
where $\beta$ denotes the shape parameter and  $\alpha$ denotes the scale parameter. Because 
\begin{equation} 
h(t)=\frac{dF(t)/dt}{1-F(t)}=\frac{p_{f}(t)}{1-F(t)} 
\end{equation}
where $F(t)$ denotes the cumulative distribution function (CDF) of $p_{f}(t)$. And 
\begin{equation}
F(t)=1-{e}^{\int_{0}^{t}h(t)dt}=1-e^{-(\frac{t}{\alpha})^{\beta}}
\end{equation}
Therefore $p_{f}(t)$ could be expressed as: 
\begin{equation}
p_{f}(t)=\frac{\beta }{\alpha }({\frac{t}{\alpha}})^{\beta-1}{e}^{-2({\frac{t}{\alpha}})^{\beta}}
\end{equation}

In \cite{54}, they determined $\alpha$ and $\beta$ via parameter estimation and gave a confidence range for $\alpha$ and $\beta$ respectively. Based upon their result, we set  $\alpha=5.4E5$ and $\beta=11$ in this paper. The failure probability, $p_{f}(t)$,  from time 0 to time  $4.5E5$ (system crash assumed) is depicted in Figure 19. Accordingly  the  entropy, $E(t)$, is demonstrated in Figure 20. From Figure 20, we observe that entropy increases monotonously during the life time of the running system. In this case, the failure probability  curve is truncated at system crash, far from the point where $p_{f}(t)=p_{w}(t)$. In some corner cases,  $p_{f}(t)$ can reach the point where $p_{f}(t)=p_{w}(t)$. However, the system  suffers from SLA violations and  restarts very soon when $p_{f}(t) > p_{w}(t)$. Thus we only take into account the scenario when $p_{f}(t) < p_{w}(t)$. In this scenario, the system entropy increases monotonously. Therefore Theorem 1 is true as long as $p_{f}(t)$ or $p_{w}(t)$ varies monotonously. 
%\begin{itemize}
%  \item When $p_{w}(t) > p_{f}(t)$, $E(t)$ increases with time;
%  \item When $p_{w}(t) = p_{f}(t)$, $E(t)$  reaches the maximum;
%  \item When $p_{w}(t) < p_{f}(t)$, $E(t)$ decreases with time;
%  \item When $t \rightarrow \infty$, $E(t)$ approximates to the minimum 0;
%\end{itemize}
\begin{figure}[!t]
\begin{minipage}[t]{0.5\linewidth}
\centering
\includegraphics[width=1.6in,height=1.2in]{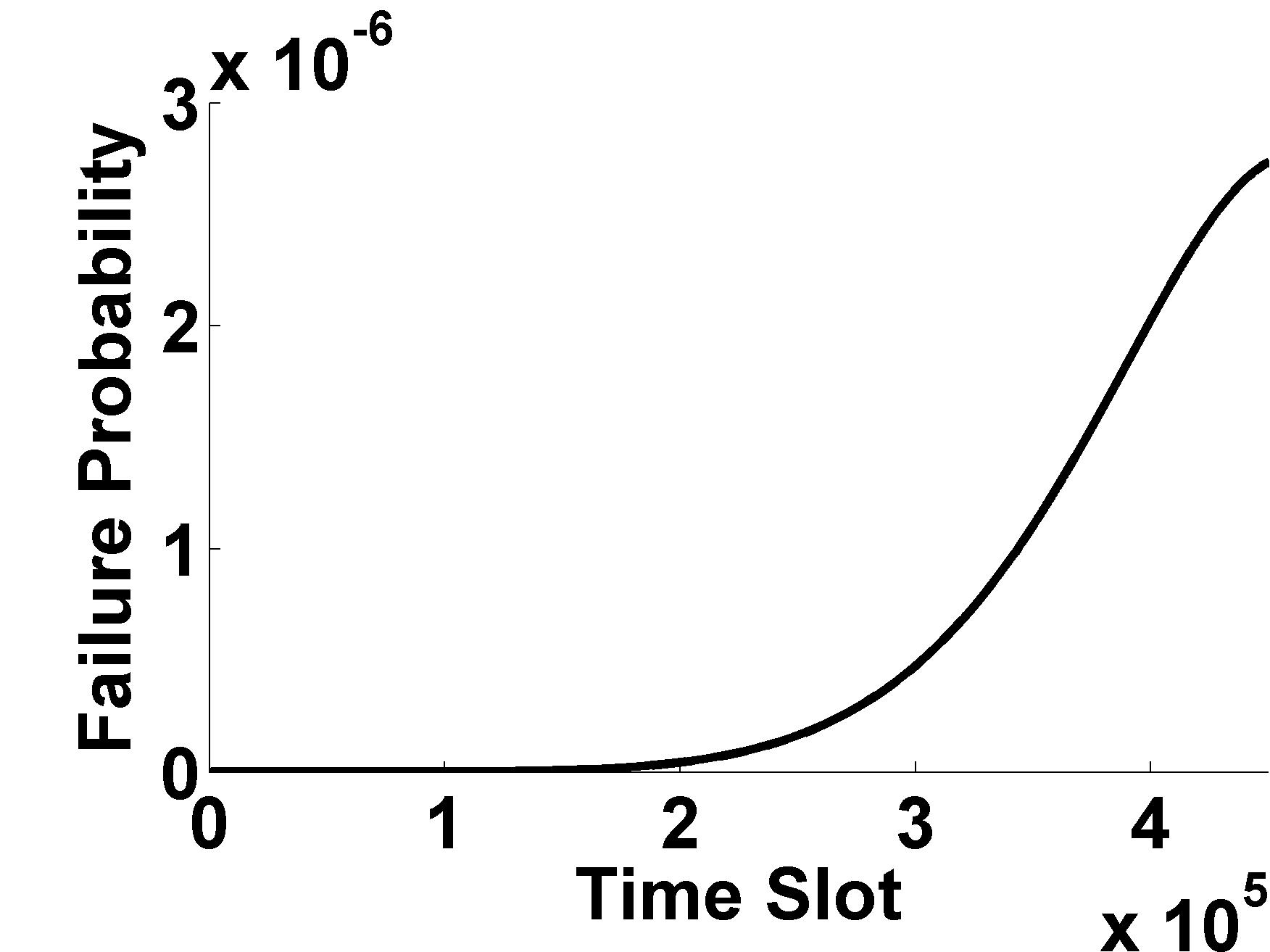}
\caption{$p_{f}(t)$ variation curve}
\label{Recall}
\end{minipage}%
\begin{minipage}[t]{0.5\linewidth}
\centering
\includegraphics[width=1.6in,height=1.2in]{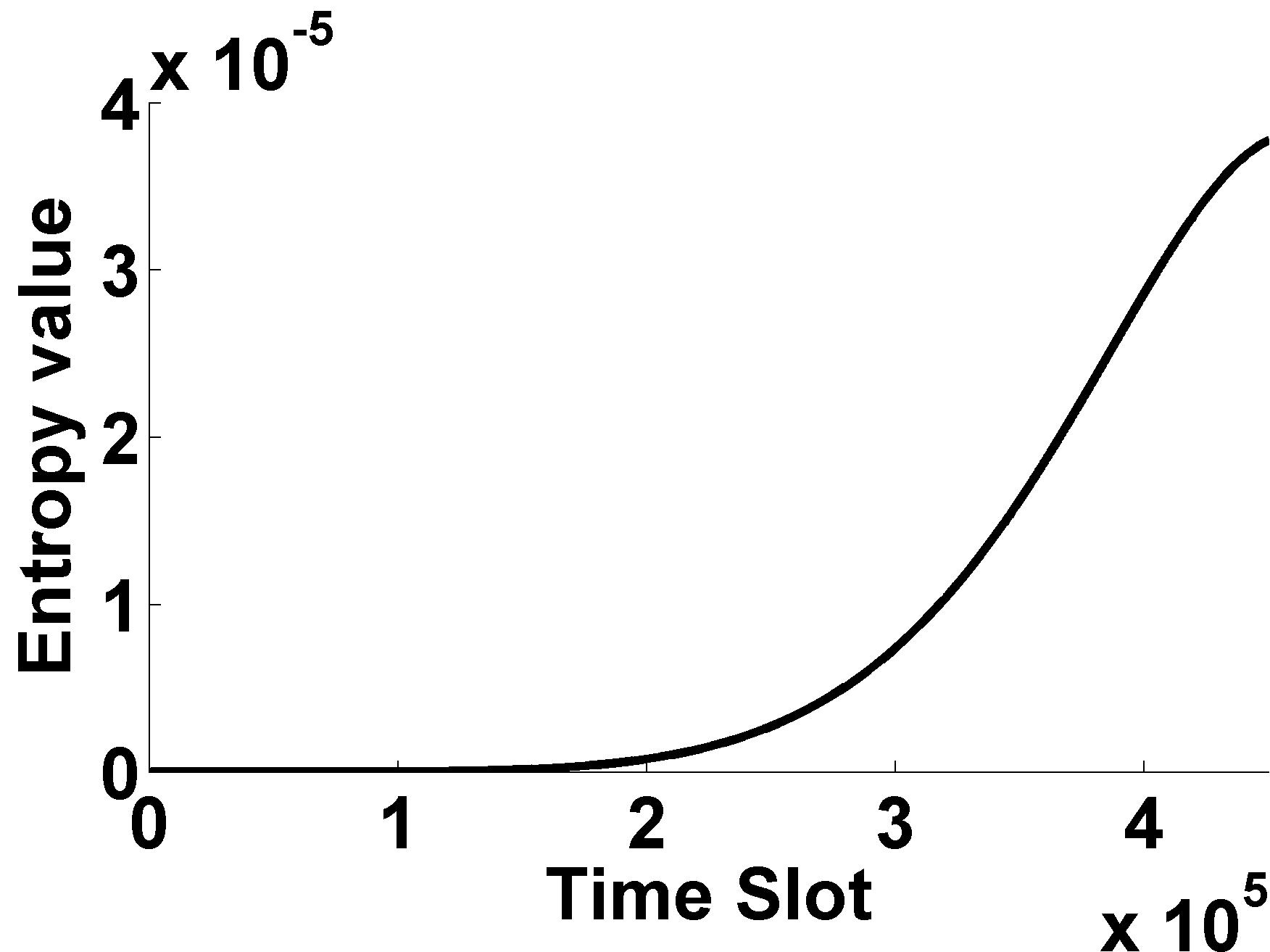}
\caption{$E(t)$ variation curve}
\label{Precision}
\end{minipage}
\end{figure}
 
% \begin{figure}[!t]
%\centering
%\includegraphics[width=3.3in]{EntropyIncrease}
%\caption{Entropy variation when $p_{f}(t)$ follows Weibull distribution.}
%\label{Entropy}
%\end{figure}
\begin{Mytheo}
 If  $p_{f}(t)$ increases monotonously, the system entropy $E(t)$ monotonously increases with the degree of software aging when $p_{f}(t) < p_{w}(t)$ or $ p_{f}(t) < \frac{1}{2}$   
\end{Mytheo}
\begin{proof}
When $p_{f}(t)=0$ or $p_{f}(t)=1$, $ln(1-p_{f}(t))$ or $ln(p_{f}(t))$ is not defined. Hence we assume $p_{f}(t) \in (0,1)$. 
 Substitute $p_{w}(t)$ with $1-p_{f}(t)$ in equation (12). Then we get:   \\
   $$ E(t)=-((1-p_{f}(t))*ln(1-p_{f}(t))+p_{f}(t)*ln(p_{f}(t)))$$
        $$\quad\quad =-ln(1-p_{f}(t))+p_{f}(t)*(ln(1-p_{f}(t))-ln(p_{f}(t)))$$ 
  Regard $p_{f}(t)$ as an variable, the first order  derivative  and second order derivative of $E(t)$ are:
  $E(t)^{'}=ln(1-p_{f}(t))-ln(p_{f}(t))$, 
 $E(t)^{''}=-((1-p_{f}(t))*p_{f}(t))^{-1}$.
 As $p_{f}(t) \in (0,1)$, $E(t)^{''}<0$. Therefore $E(t)$ achieves the maximum value when  $E(t)^{'}=0$ namely
 $ln(p_{f}(t))=ln(1-p_{f}(t))$
Finally, we get $p_{f}(t)=\frac{1}{2}$. As $p_{f}(t)$ increases monotonously, $E(t)$ increases monotonously when  $p_{f}(t)<\frac{1}{2}$. Hence Theorem 1 is proved.
\end{proof}
\ifCLASSOPTIONcompsoc
  \section*{Acknowledgments}
\else
  \section*{Acknowledgment}
\fi

The authors would like to thank all the members in our research group and the anonymous reviewers.

\ifCLASSOPTIONcaptionsoff
  \newpage
\fi


%\bibliographystyle{IEEEtran}
%\begin{thebibliography}{1}
%H.~Kopka and P.~W. Daly, \emph{A Guide to {\LaTeX}}, 3rd~ed.\hskip 1em plus
%\end{thebibliography}
\bibliographystyle{IEEEtran}
\bibliography{IEEEfull,reference}

\begin{thebibliography}{10}
\providecommand{\url}[1]{#1}
\csname url@samestyle\endcsname
\providecommand{\newblock}{\relax}
\providecommand{\bibinfo}[2]{#2}
\providecommand{\BIBentrySTDinterwordspacing}{\spaceskip=0pt\relax}
\providecommand{\BIBentryALTinterwordstretchfactor}{4}
\providecommand{\BIBentryALTinterwordspacing}{\spaceskip=\fontdimen2\font plus
\BIBentryALTinterwordstretchfactor\fontdimen3\font minus
  \fontdimen4\font\relax}
\providecommand{\BIBforeignlanguage}[2]{{%
\expandafter\ifx\csname l@#1\endcsname\relax
\typeout{** WARNING: IEEEtran.bst: No hyphenation pattern has been}%
\typeout{** loaded for the language `#1'. Using the pattern for}%
\typeout{** the default language instead.}%
\else
\language=\csname l@#1\endcsname
\fi
#2}}
\providecommand{\BIBdecl}{\relax}
\BIBdecl

\bibitem{0}
K.~Vaidyanathan and K.~S. Trivedi, ``A measurement-based model for estimation
  of resource exhaustion in operational software systems,'' in \emph{Software
  Reliability Engineering, 1999. Proceedings. 10th International Symposium
  on}.\hskip 1em plus 0.5em minus 0.4em\relax IEEE, 1999, pp. 84--93.

\bibitem{1}
K.~Vaidyanathan, R.~E. Harper, S.~W. Hunter, and K.~S. Trivedi, ``Analysis and
  implementation of software rejuvenation in cluster systems,'' in \emph{ACM
  SIGMETRICS Performance Evaluation Review}, vol.~29, no.~1.\hskip 1em plus
  0.5em minus 0.4em\relax ACM, 2001, pp. 62--71.

\bibitem{2}
K.~Vaidyanathan and K.~S. Trivedi, ``A comprehensive model for software
  rejuvenation,'' \emph{Dependable and Secure Computing, IEEE Transactions on},
  vol.~2, no.~2, pp. 124--137, 2005.

\bibitem{3}
M.~Grottke, L.~Li, K.~Vaidyanathan, and K.~S. Trivedi, ``Analysis of software
  aging in a web server,'' \emph{Reliability, IEEE Transactions on}, vol.~55,
  no.~3, pp. 411--420, 2006.

\bibitem{4}
J.~Alonso, J.~Torres, J.~L. Berral, and R.~Gavalda, ``Adaptive on-line software
  aging prediction based on machine learning,'' in \emph{Dependable Systems and
  Networks (DSN), 2010 IEEE/IFIP International Conference on}.\hskip 1em plus
  0.5em minus 0.4em\relax IEEE, 2010, pp. 507--516.

\bibitem{5}
S.~P. Kavulya, S.~Daniels, K.~Joshi, M.~Hiltunen, R.~Gandhi, and P.~Narasimhan,
  ``Draco: Statistical diagnosis of chronic problems in large distributed
  systems,'' in \emph{Dependable Systems and Networks (DSN), 2012 42nd Annual
  IEEE/IFIP International Conference on}.\hskip 1em plus 0.5em minus
  0.4em\relax IEEE, 2012, pp. 1--12.

\bibitem{6}
Y.~Huang, C.~Kintala, N.~Kolettis, and N.~D. Fulton, ``Software rejuvenation:
  Analysis, module and applications,'' in \emph{Fault-Tolerant Computing, 1995.
  FTCS-25. Digest of Papers., Twenty-Fifth International Symposium on}.\hskip
  1em plus 0.5em minus 0.4em\relax IEEE, 1995, pp. 381--390.

\bibitem{7}
J.~Araujo, R.~Matos, V.~Alves, P.~Maciel, F.~Souza, K.~S. Trivedi
  \emph{et~al.}, ``Software aging in the eucalyptus cloud computing
  infrastructure: Characterization and rejuvenation,'' \emph{ACM Journal on
  Emerging Technologies in Computing Systems (JETC)}, vol.~10, no.~1, p.~11,
  2014.

\bibitem{8}
J.~Araujo, R.~Matos, P.~Maciel, R.~Matias, and I.~Beicker, ``Experimental
  evaluation of software aging effects on the eucalyptus cloud computing
  infrastructure,'' in \emph{Proceedings of the Middleware 2011 Industry Track
  Workshop}.\hskip 1em plus 0.5em minus 0.4em\relax ACM, 2011, p.~4.

\bibitem{10}
K.~Kourai and S.~Chiba, ``Fast software rejuvenation of virtual machine
  monitors,'' \emph{Dependable and Secure Computing, IEEE Transactions on},
  vol.~8, no.~6, pp. 839--851, 2011.

\bibitem{11}
------, ``A fast rejuvenation technique for server consolidation with virtual
  machines,'' in \emph{Dependable Systems and Networks, 2007. DSN'07. 37th
  Annual IEEE/IFIP International Conference on}.\hskip 1em plus 0.5em minus
  0.4em\relax IEEE, 2007, pp. 245--255.

\bibitem{12}
D.~Cotroneo, R.~Natella, R.~Pietrantuono, and S.~Russo, ``Software aging
  analysis of the linux operating system,'' in \emph{Software Reliability
  Engineering (ISSRE), 2010 IEEE 21st International Symposium on}.\hskip 1em
  plus 0.5em minus 0.4em\relax IEEE, 2010, pp. 71--80.

\bibitem{13}
D.~Cotroneo, S.~Orlando, and S.~Russo, ``Characterizing aging phenomena of the
  java virtual machine,'' in \emph{Reliable Distributed Systems, 2007. SRDS
  2007. 26th IEEE International Symposium on}.\hskip 1em plus 0.5em minus
  0.4em\relax IEEE, 2007, pp. 127--136.

\bibitem{14}
Y.-F. Jia, L.~Zhao, and K.-Y. Cai, ``A nonlinear approach to modeling of
  software aging in a web server,'' in \emph{Software Engineering Conference,
  2008. APSEC'08. 15th Asia-Pacific}.\hskip 1em plus 0.5em minus 0.4em\relax
  IEEE, 2008, pp. 77--84.

\bibitem{15}
B.~Sharma, P.~Jayachandran, A.~Verma, and C.~R. Das, ``Cloudpd: Problem
  determination and diagnosis in shared dynamic clouds,'' in \emph{IEEE DSN},
  2013.

\bibitem{17}
V.~Castelli, R.~E. Harper, P.~Heidelberger, S.~W. Hunter, K.~S. Trivedi,
  K.~Vaidyanathan, and W.~P. Zeggert, ``Proactive management of software
  aging,'' \emph{IBM Journal of Research and Development}, vol.~45, no.~2, pp.
  311--332, 2001.

\bibitem{19}
P.~Zheng, Y.~Qi, Y.~Zhou, P.~Chen, J.~Zhan, and M.~Lyu, ``An automatic
  framework for detecting and characterizing performance degradation of
  software systems,'' \emph{Reliability, IEEE Transactions on}, vol.~63, no.~4,
  pp. 927--943, 2014.

\bibitem{20}
S.~Garg, A.~van Moorsel, K.~Vaidyanathan, and K.~S. Trivedi, ``A methodology
  for detection and estimation of software aging,'' in \emph{Software
  Reliability Engineering, 1998. Proceedings. The Ninth International Symposium
  on}.\hskip 1em plus 0.5em minus 0.4em\relax IEEE, 1998, pp. 283--292.

\bibitem{24}
M.~Shereshevsky, J.~Crowell, B.~Cukic, V.~Gandikota, and Y.~Liu, ``Software
  aging and multifractality of memory resources,'' in \emph{2003 33rd Annual
  IEEE/IFIP International Conference on Dependable Systems and Networks
  (DSN)}.\hskip 1em plus 0.5em minus 0.4em\relax IEEE Computer Society, 2003,
  pp. 721--730.

\bibitem{28}
I.~Jolliffe, \emph{Principal component analysis}.\hskip 1em plus 0.5em minus
  0.4em\relax Wiley Online Library, 2005.

\bibitem{30}
J.~F. Cadima and I.~T. Jolliffe, ``Variable selection and the interpretation of
  principal subspaces,'' \emph{Journal of agricultural, biological, and
  environmental statistics}, vol.~6, no.~1, pp. 62--79, 2001.

\bibitem{31}
J.~Ramsay, J.~ten Berge, and G.~Styan, ``Matrix correlation,''
  \emph{Psychometrika}, vol.~49, no.~3, pp. 403--423, 1984.

\bibitem{29}
J.~Cadima, J.~O. Cerdeira, and M.~Minhoto, ``Computational aspects of
  algorithms for variable selection in the context of principal components,''
  \emph{Computational statistics \& data analysis}, vol.~47, no.~2, pp.
  225--236, 2004.

\bibitem{35}
C.~E. Shannon, ``Bell system tech. j. 27 (1948) 379; ce shannon,'' \emph{Bell
  System Tech. J}, vol.~27, p. 623, 1948.

\bibitem{27}
M.~Costa, A.~L. Goldberger, and C.-K. Peng, ``Multiscale entropy analysis of
  biological signals,'' \emph{Physical Review E}, vol.~71, no.~2, p. 021906,
  2005.

\bibitem{59}
S.~M. Pincus and A.~L. Goldberger, ``Physiological time-series analysis: what
  does regularity quantify?'' \emph{American Journal of Physiology}, vol. 266,
  pp. H1643--H1643, 1994.

\bibitem{62}
K.~C. Gross, V.~Bhardwaj, and R.~Bickford, ``Proactive detection of software
  aging mechanisms in performance critical computers,'' in \emph{Software
  Engineering Workshop, 2002. Proceedings. 27th Annual NASA
  Goddard/IEEE}.\hskip 1em plus 0.5em minus 0.4em\relax IEEE, 2002, pp. 17--23.

\bibitem{37}
M.~U. Ahmed and D.~P. Mandic, ``Multivariate multiscale entropy: A tool for
  complexity analysis of multichannel data,'' \emph{Physical Review E},
  vol.~84, no.~6, p. 061918, 2011.

\bibitem{38}
L.~Cao, A.~Mees, and K.~Judd, ``Dynamics from multivariate time series,''
  \emph{Physica D: Nonlinear Phenomena}, vol. 121, no.~1, pp. 75--88, 1998.

\bibitem{18}
D.~Cotroneo, R.~Natella, R.~Pietrantuono, and S.~Russo, ``A survey of software
  aging and rejuvenation studies,'' \emph{ACM Journal on Emerging Technologies
  in Computing Systems (JETC)}, vol.~10, no.~1, p.~8, 2014.

\bibitem{61}
K.~J. Cassidy, K.~C. Gross, and A.~Malekpour, ``Advanced pattern recognition
  for detection of complex software aging phenomena in online transaction
  processing servers,'' in \emph{Dependable Systems and Networks, 2002. DSN
  2002. Proceedings. International Conference on}.\hskip 1em plus 0.5em minus
  0.4em\relax IEEE, 2002, pp. 478--482.

\bibitem{26}
M.~Costa, A.~L. Goldberger, and C.-K. Peng, ``Multiscale entropy analysis of
  complex physiologic time series,'' \emph{Physical review letters}, vol.~89,
  no.~6, p. 068102, 2002.

\bibitem{22}
L.~M. Silva, J.~Alonso, and J.~Torres, ``Using virtualization to improve
  software rejuvenation,'' \emph{IEEE Transactions on Computers}, vol.~58,
  no.~11, pp. 1525--1538, 2009.

\bibitem{23}
J.~Alonso, {\'I}.~Goiri, J.~Guitart, R.~Gavalda, and J.~Torres, ``Optimal
  resource allocation in a virtualized software aging platform with software
  rejuvenation,'' in \emph{Software Reliability Engineering (ISSRE), 2011 IEEE
  22nd International Symposium on}.\hskip 1em plus 0.5em minus 0.4em\relax
  IEEE, 2011, pp. 250--259.

\bibitem{46}
\BIBentryALTinterwordspacing
 [Online]. Available: \url{http://www.sourceforge.net/projects/hyperic-hq}
\BIBentrySTDinterwordspacing

\bibitem{40}
P.~Chen, Y.~Qi, P.~Zheng, J.~Zhan, and Y.~Wu, ``Multi-scale entropy: One metric
  of software aging,'' in \emph{Service Oriented System Engineering (SOSE),
  2013 IEEE 7th International Symposium on}.\hskip 1em plus 0.5em minus
  0.4em\relax IEEE, 2013, pp. 162--169.

\bibitem{41}
P.~Zheng, Y.~Zhou, M.~R. Lyu, and Y.~Qi, ``Granger causality-aware prediction
  and diagnosis of software degradation,'' in \emph{Services Computing (SCC),
  2014 IEEE International Conference on}.\hskip 1em plus 0.5em minus
  0.4em\relax IEEE, 2014, pp. 528--535.

\bibitem{47}
\BIBentryALTinterwordspacing
 [Online]. Available: \url{http://www.helix-server.helixcommunity.org/}
\BIBentrySTDinterwordspacing

\bibitem{50}
A.~Andrzejak and L.~Silva, ``Using machine learning for non-intrusive modeling
  and prediction of software aging,'' in \emph{Network Operations and
  Management Symposium, 2008. NOMS 2008. IEEE}.\hskip 1em plus 0.5em minus
  0.4em\relax IEEE, 2008, pp. 25--32.

\bibitem{60}
R.~Matias, P.~A. Barbetta, K.~S. Trivedi \emph{et~al.}, ``Accelerated
  degradation tests applied to software aging experiments,'' \emph{Reliability,
  IEEE Transactions on}, vol.~59, no.~1, pp. 102--114, 2010.

\bibitem{21}
J.~P. Magalhaes and L.~Moura~Silva, ``Prediction of performance anomalies in
  web-applications based-on software aging scenarios,'' in \emph{Software Aging
  and Rejuvenation (WoSAR), 2010 IEEE Second International Workshop on}.\hskip
  1em plus 0.5em minus 0.4em\relax IEEE, 2010, pp. 1--7.

\bibitem{44}
M.~U. Ahmed and D.~P. Mandic, ``Multivariate multiscale entropy analysis,''
  \emph{Signal Processing Letters, IEEE}, vol.~19, no.~2, pp. 91--94, 2012.

\bibitem{43}
C.~Wang, V.~Talwar, K.~Schwan, and P.~Ranganathan, ``Online detection of
  utility cloud anomalies using metric distributions,'' in \emph{Network
  Operations and Management Symposium (NOMS), 2010 IEEE}.\hskip 1em plus 0.5em
  minus 0.4em\relax IEEE, 2010, pp. 96--103.

\bibitem{54}
J.~Zhao, Y.~Jin, K.~S. Trivedi, and R.~Matias, ``Injecting memory leaks to
  accelerate software failures,'' in \emph{Software Reliability Engineering
  (ISSRE), 2011 IEEE 22nd International Symposium on}.\hskip 1em plus 0.5em
  minus 0.4em\relax IEEE, 2011, pp. 260--269.

\bibitem{55}
K.~S. Trivedi, K.~Vaidyanathan, and K.~Goseva-Popstojanova, ``Modeling and
  analysis of software aging and rejuvenation,'' in \emph{Simulation Symposium,
  2000.(SS 2000) Proceedings. 33rd Annual}.\hskip 1em plus 0.5em minus
  0.4em\relax IEEE, 2000, pp. 270--279.

\bibitem{45}
A.~Bobbio, M.~Sereno, and C.~Anglano, ``Fine grained software degradation
  models for optimal rejuvenation policies,'' \emph{Performance Evaluation},
  vol.~46, no.~1, pp. 45--62, 2001.

\bibitem{51}
Y.~Bao, X.~Sun, and K.~S. Trivedi, ``A workload-based analysis of software
  aging, and rejuvenation,'' \emph{Reliability, IEEE Transactions on}, vol.~54,
  no.~3, pp. 541--548, 2005.

\bibitem{52}
T.~Dohi, K.~Goseva-Popstojanova, and K.~S. Trivedi, ``Analysis of software cost
  models with rejuvenation,'' in \emph{High Assurance Systems Engineering,
  2000, Fifth IEEE International Symposim on. HASE 2000}.\hskip 1em plus 0.5em
  minus 0.4em\relax IEEE, 2000, pp. 25--34.

\bibitem{53}
R.~E. Barlow and R.~A. Campo, ``Total time on test processes and applications
  to failure data analysis.'' DTIC Document, Tech. Rep., 1975.

\bibitem{48}
T.~Dohi, K.~Goseva-Popstojanova, and K.~S. Trivedi, ``Statistical
  non-parametric algorithms to estimate the optimal software rejuvenation
  schedule,'' in \emph{Dependable Computing, 2000. Proceedings. 2000 Pacific
  Rim International Symposium on}.\hskip 1em plus 0.5em minus 0.4em\relax IEEE,
  2000, pp. 77--84.

\end{thebibliography}
\end{document}